\pdfoutput=1
\documentclass{amsart}
\usepackage{graphicx,euscript,verbatim,constants}
\usepackage{amssymb,url,booktabs}

\newconstantfamily{const}{symbol=c}
\renewcommand{\c}[1]{\Cl[const]{#1}}
\newcommand{\nc}{\newcommand}
\newtheorem{Thm}{Theorem}
\newtheorem{Cor}[Thm]{Corollary}
\newtheorem{Claim}[Thm]{Claim}
\newtheorem{lemma}[Thm]{Lemma}

\nc{\wt}{\widetilde}
\nc{\izf}{\int_{0}^{\infty}}
\nc{\constc}{c}
\nc{\constcp}{c'}
\nc{\J}{\eu{J}}
\nc{\tB}{\wt{B}}
\nc{\tf}{\tilde{f}}
\nc{\tD}{\wt{\Delta}}
\nc{\U}{\mathcal{U}}
\nc{\V}{\mathcal{V}}
\nc{\tU}{\check{U}}
\nc{\tV}{\check{V}}
\nc{\tN}{\wt{N}}
\nc{\tY}{\wt{Y}}
\nc{\R}{\mathbb{R}}
\nc{\F}{\eu{F}}
\newcommand{\tX}{\wt{X}}

\newcommand{\E}{\mathbb{E}}
\normalsize

\nc{\tP}{\wt{P}}
\renewcommand{\P}{\mathbb{P}}
\nc{\tq}{\tilde{q}}
\nc{\tZ}{\wt{Z}}
\nc{\tW}{\wt{W}}
\nc{\Xps}{X^{(\epsilon)}}
\nc{\Yps}{Y^{(\epsilon)}}
\nc{\aps}{a^{(\epsilon)}}
\nc{\e}{\operatorname{e}}
\nc{\bX}{\mathbf{X}}
\nc{\tmu}{\tilde{\mu}}
\nc{\tx}{\tilde{\xi}}
\nc{\ttau}{\tilde{\tau}}
\nc{\tsig}{\tilde{\sigma}}
\nc{\hA}{\hat{A}}
\nc{\hR}{\hat{R}}
\nc{\ha}{\hat{a}}
\nc{\hD}{\hat{D}^{*}}
\nc{\bA}{\mathbf{A}}

\nc{\M}{\eu{M}}
\nc{\eu}{\EuScript}
\nc{\cond}{\, \bigm|\,}
\newcommand{\convinfty}[1]
  {\stackrel{ #1\to\infty}{-\hspace{-2mm}-\hspace{-2mm}
  -\hspace{-2mm}-\hspace{-4mm}\longrightarrow}}
\newcommand{\indic}{\boldsymbol{1}}

\begin{document}
\title[Growth rates with rare migration]{Stochastic growth rates for life histories with rare migration or diapause}
\author{David Steinsaltz} 
\author{Shripad Tuljapurkar}
\address{David Steinsaltz\\Department of Statistics\\University of Oxford\\1 South Parks Road\\Oxford OX1 3TG\\United Kingdom}
 \address{Shripad Tuljapurkar\\454 Herrin Labs\\Department of Biology\\Stanford University\\Stanford CA 94305-5020\\USA}

\begin{abstract}
The growth of a population divided among spatial sites, with migration between the sites, is sometimes modelled by
a product of random matrices, with each diagonal elements representing the growth rate in a given time period, and 
off-diagonal elements the migration rate. If the sites are reinterpreted as age classes, the same model may apply to
a single population with age-dependent mortality and reproduction.

We consider the case where the off-diagonal elements are small, representing a situation where there is little migration or,
alternatively, where a deterministic life-history has been slightly disrupted, for example by introducing a rare delay in development.
We examine the asymptotic behaviour of the long-term growth rate. We show that when the highest growth rate is attained at two
different sites in the absence of migration
(which is always the case when modelling a single age-structured population) the increase in stochastic growth rate due to a migration
rate $\epsilon$ is like $(\log \epsilon^{-1})^{-1}$ as $\epsilon\downarrow 0$, under fairly generic conditions. When there is a single site
with the highest growth rate the behavior is more delicate, depending on the tails of the growth rates. For the case when the log growth rates have Gaussian-like tails we show that the behavior near
zero is like a power of $\epsilon$, and derive upper and lower bounds for the power in terms of the difference in the growth rates
and the distance between the sites.
\end{abstract}
\maketitle

\section{Introduction}
\subsection{Biological motivation}
If a population is divided among spatial sites, with distinct fixed growth rates at each site, 
with no migration between sites, the numbers in the best site will become overwhelmingly larger than
those at the other sites, and the overall population growth rate will be determined by 
the rate prevailing at the best site. Introducing migration between sites, as Karlin showed \cite{karlin82},
will always reduce the long-run growth rate of the total population.

Karlin's theorem assumes deterministic growth. den Boer \cite{boer1968spreading} argued that migration may increase long-run growth when there is independent or weakly correlated stochastic variation in growth among sites. In a different context, it has long been argued \cite{cole1954pcl} that populations of individuals who delay or spread reproduction over time will suffer reduced growth rate. 
But Cohen \cite{cohen1966} and Cohen and Levin \cite{cohen1991} used analysis and simulations to show that long-run growth of a population could increase as a result of a life cycle delay when there are some kinds of random variation in time, or by migration when there are some kinds of random variation across space. 
These kinds of stochastic variation have been formulated as random matrix models  whose Lyapunov exponent  is the long-run growth rate of the population, as discussed by \cite{tuljapurkar2000ets, Wiener1994}. In this general setting, we would like to know whether the long-run growth rate increases when there is mixing in space and/or time \cite{tuljapurkar2000ets} --- biologically, when should migration and/or delay be favoured to
evolve? A general and precise answer has been difficult because previous work \cite{Wiener1994} shows that the long-run growth rate can be singular (e.g., non-differentiable) in the limit of no mixing. A similar singularity arises in random-matrix models used in models of disordered matter (\cite{derrida1983singular}). 

Here we consider a random-matrix model of migration among sites whose individual growth rates vary stochastically over time, and characterize the behavior of the Lyapunov exponent in the limit of zero migration. 
As we show, this model can  be used to study a number of models of migration, life cycle delay, 
or a combination of these. Our results address evolutionary stability (in a fitness-maximising context) of a small amount of mixing, via migration or life-cycle delays. 
We find that
when the growth rates at different sites are similar, the population with
a small positive migration rate will profit from the variability between sites,
to grow faster than would have been possible even at the best single site.
In particular, when the growth rates are equal --- which is always the case
in the diapause setting --- the sensitivity of stochastic growth rate to changes
migration rate is extreme, varying near 0 like $1/\log \epsilon^{-1}$.

Our results complement a recent analysis of spatially stochastic growth in \cite{ERSS12}. They use diffusions to model migration among sites that have independently varying stochastic growth rates and characterize the migration rate that maximizes the long-run stochastic growth rate. 

\subsection{The mathematical problem: Migration} \label{sec:migration}
Suppose $D_{1},D_{2},\dots$ is an i.i.d.\  sequence of $d\times d$ diagonal matrices.
We write $\xi_{t}^{(0)},\dots,\xi_{t}^{(d-1)}$ for the diagonal elements of
$D_{t}$, and assume $X_{t}^{(i)}:=\log\xi_{t}^{(i)}$ all have finite 
mean $\mu_{i}$ and  finite variance. We order them so that $\mu_{0}$ 
is the largest. We also write $\tmu^{(i)}:=\mu_{0}-\mu_{i}$. 

We let $A_{t}$ be an i.i.d.\ sequence of nonnegative
$d\times d$ matrices with zeros on the diagonal, and we assume that
the pairs $(D_{t},A_{t})$ are all independent. 
We define the {\em migration graph} $\M$ to be a directed graph 
whose vertices are the sites $\{0,\dots,d-1\}$,
possessing an edge $(i,j)$ if $\operatorname{ess}\inf \E[\log A_{t}(j,i)\cond D_{t}]>-\infty$.
That is, we do not assume that $A_{t}$ and $D_{t}$ are independent, but for the migration
steps we consider to be possible we assume that there is a lower bound to how close $A_{t}(j,i)$ can
come to 0 that is independent of all $D_{t}$. It follows trivially that when $(i,j)$ is an
edge $\E[\log A_{t}(j,i)\cond \eu{D}]>-\infty$
where $\eu{D}$ is the $\sigma$-algebra generated by $\{D_{t},\, 1\le t<\infty\}$. 
We assume that $\M$ is connected.

We let  $\Delta$ be
a fixed diagonal matrix with entries $\Delta_{0},\dots,\Delta_{d-1}$. (Generally we will be thinking of $\Delta$ as the growth or survival
penalty for migration or diapause, so that the entries will be negative, but this
is not essential.) We assume the penalty acts multiplicatively on growth ---
this seems reasonable from a modeling perspective, and elegantly avoids the 
problem of negative matrix entries --- and is proportional to $\epsilon$.   We define
$$
D_{t}(\epsilon):= \e^{\epsilon \Delta}D_{t}+\epsilon A_{t}.
$$
We also write $\tD^{(i)}:=\Delta_{0}-\Delta_{i}$.

For $\epsilon>0$ the i.i.d.\  sequence $D_{t}(\epsilon)$
satisfies the conditions for the existence of a stochastic growth rate independent
of starting condition.\cite{jC79} That is, if we define the partial products
$$
R_{T}(\epsilon):= D_{T}(\epsilon)\cdot D_{T-1}(\epsilon)\cdot\cdots\cdot D_{1}(\epsilon)
$$
then
$$
a(\epsilon):=\lim_{T\to\infty} T^{-1}\log R_{T}(\epsilon)_{i,j}.
$$
are well defined deterministic quantities, in the sense that the limit exists
almost surely, is almost-surely constant, and is the same for any 
$0\le i,j\le d-1$.

Of course, $R_{T}(0)$ is not so simple. The off-diagonal terms
are all 0, while on the diagonal, by the Strong Law of Large Numbers,
$$
\lim_{T\to\infty} T^{-1}\log R_{T}(0)_{i,i}=\mu_{i}.
$$
We take $a(0):=\mu_{0}=\lim_{\epsilon\downarrow 0}a(\epsilon)$. (That is,
the growth rate of all matrix entries when $\epsilon>0$ is small will be
dominated by the largest growth rate in a single diagonal entry.
This will be an elementary consequence of the stronger results
that will be stated in section \ref{sec:mainresult}, and proved in the sequel.)

\subsection{Variation of the mathematical problem: Diapause}  \label{sec:diapause}
Consider a population in which individuals progress through immature life stages until reaching adulthood, when they reproduce and
then die. Diapause is a life-cycle delay in which individuals can stay in some immature stage with some probability. We can describe diapause
by reconceptualizing the ``sites'' of the previous section as life stages, and also describe an organism's progress using matrices that are not diagonal, but sub-diagonal. The life stages  (or sites)
are viewed as a cycle, described by matrices of the form$$
M_{t}:= 
\begin{pmatrix}
0&0&\cdots&0&B_{t}\\
S_{t}^{(0)}&0&\cdots&0&0\\
0&S_{t}^{(1)}&\cdots&0&0\\
\vdots&\vdots&\ddots&\vdots&\vdots\\
0&0&\cdots&S_{t}^{(d-2)}&0
\end{pmatrix}
$$
Here ages run from $1$ to $d$, and are equivalently referred to as age classes that run from 0 to $d-1$. The quantity $S_{t}^{(j)}\in (0,1)$ is the proportion surviving from age $j$ to $j+1$ in year $t$, and $B_{t}$ is the average number of offspring produced when an individual becomes mature in age-class $d-1$.
Offspring are born into age-class 0, and the parent --- in age class $d-1$ --- dies.To this we add $\epsilon A$, where now
$A$ is a fixed diagonal matrix with nonnegative entries, and at least one positive entry, and also allow for
penalties $\e^{-\epsilon\Delta_{i}}$.

We immediately have
\begin{equation} \label{E:a0diapause}
a(0)=\frac{1}{d}\Bigl(\E[\log B_{t}]+\sum_{j=0}^{d-2} \E[\log S_{t}^{(j)}]\Bigr).
\end{equation}
If we look at this in groups of $d$ generations, the product
$$
D_{t}:=\bigr(\e^{\epsilon \Delta}M_{dt+d-1}+\epsilon A_{dt+d-1}\bigl)\bigr(\e^{\epsilon \Delta}M_{dt+d-2}+\epsilon A_{dt+d-2}\bigl)\cdots \bigr(\e^{\epsilon \Delta}M_{dt}+\epsilon A_{dt}\bigl)
$$
is diagonal when $\epsilon=0$, and is of the form described in section \ref{sec:migration}.

Consequently, we may apply Theorem \ref{T:samerate1} to this $D_{t}$, producing the same
$1/\log\epsilon^{-1}$ rate of increase, as stated in Corollary \ref{C:diapause}.

\subsection{The effect of the penalty $\Delta$}
We will mostly be concerned with analyzing the case $\Delta\equiv 0$.
For most purposes, $\Delta$ has no effect. 
But this is not always true.

The crucial point is that the effect of $\Delta$ is always nearly linear in $\epsilon$,
while the increase of $a$ near 0 is often superlinear,
growing either as $1/\log\epsilon^{-1}$ or as $\epsilon^{\beta}$ for a power
$\beta<1$. In either case, the rapid increase in $a$ near 0 will be qualitatively
unaffected by a linear term for $\epsilon$ sufficiently small. 
Even when the linear term is negative (as we will generally be assuming it to be),
the growth rate $a$ will still be increasing on a small interval of $\epsilon>0$.

On the other hand, as discussed in section \ref{sec:mainresult} in some cases
we cannot exclude the possibility that the growth rate when $\Delta\equiv 0$
is qualitatively like $\epsilon^{\beta}$ with $\beta\ge 1$. If $\beta>1$ and
$\Delta_{0}<0$ then $a$ will be decreasing near 0; if $\beta=1$ then a
more sensitive analysis would be required.

Since the upper and lower bounds on the appropriate power of $\epsilon$ in 
Theorem \ref{T:diffrate1} are distinct, 
with the lower bound on the growth rate (the upper bound on the power of 
$\epsilon$) being sometimes larger than 1, it will not always be possible to determine
whether the change in growth rate increases or decreases for infinitesimal changes
in $\epsilon$.

\subsection{Orlicz norm and sub-Gaussian variance factor} \label{sec:Orlicz}
Let $\Psi(x)=e^{x^{2}}/5$. Following \cite{dP90} we define the Orlicz norm 
$\|Z\|_{\Psi}$ for a centered random variable $Z$ by
\begin{equation} \label{E:orlicz}
\|Z\|_\Psi:=\inf\{C\,:\, \E[\Psi(|Z|/C)]<1\}.
\end{equation}

In \cite{BLM13} a random variable $Z$ is said to be {\em sub-Gaussian} if it
has finite {\em variance factor} $\tau(Z)$, defined as
\begin{equation} \label{E:sGs}
\tau(Z):=\inf \left\{ a\ge 0\,:\, \E\left[\e^{\lambda Z}\right] \le 
   \e^{a\lambda^{2}/2} \, \forall \lambda \in\R\right\}.
\end{equation}
(The square-root of this 
is called the {\em sub-Gaussian standard} in \cite{BK00}.)
We also introduce the term {\em subvariance} to denote
\begin{equation} \label{E:subvar}
\tau_{*}(Z):= \liminf_{z\to\infty}\frac{z^{2}}{-2\log \P\bigl\{ Z>z\bigr\}}.
\end{equation}
Note that the variance factor may be thought of as an upper 
bound on the scale of the tails, while the subvariance is a lower
bound. When $Z$ is not centered (that is, $\E[Z]\ne 0$) we apply the definition
to $Z-\E[Z]$.

We point out here that the assumption that $\tX_{t}=\log(\xi_{t}^{(1)}/\xi_{t}^{(0)})$ 
have nonzero subvariance implies what may be considered exceptionally heavy tails
for the growth rates --- effectively, something like log-normal. 
This is what is required for a nontrivial lower bound in Theorem \ref{T:diffrate1}.
Thus, it seems plausible to infer that the population will obtain no long-term benefit
from sending occasional individuals to a site with lower average growth, unless
the low average growth is compensated by fat positive tails, meaning that there is a small chance
of a very large payoff. (These heavy tails may also be generated if $\xi_{t}^{(0)}$
puts too much probability near 0 --- that is, a population crash.)

\begin{lemma}  \label{L:sGs}
A sub-Gaussian centered random variable $Z$ satisfies
\begin{align}
 \label{E:orliczsG}
\|Z\|_{\Psi}&\le \sqrt{\frac{5 \tau(Z)}{2}};\\
\tau_{*}(Z)&\le \tau(Z)<\infty. \label{E:comparevar}
\end{align}

If $Z$ is Gaussian with mean 0 and variance $\sigma^{2}$ then 
\begin{equation} \label{E:gausssubgauss}
\tau_{*}(Z)=\tau(Z)=\sigma^{2}.
\end{equation}
\end{lemma}

\begin{proof}
The statement \eqref{E:gausssubgauss} is trivial.

If $\tau=\tau(Z)$ is finite then for any $ \lambda,z,\delta>0$,
$$
\P\bigl\{ |Z|>z \bigr\} \le \e^{\frac{(\tau+\delta)\lambda^{2}}{2}- \lambda z } .
$$
Taking $\lambda=z/(\tau+\delta)$, we have $\P\bigl\{ |Z|>z \bigr\} \le \e^{-z^{2}/2(\tau+\delta)}$,
which implies
\begin{equation} \label{E:taubound}
\P\bigl\{ |Z|>z \bigr\} \le \e^{-z^{2}/2\tau},
\end{equation}
since $\delta$ is arbitrary.
This immediately proves \eqref{E:comparevar}.

Integrating by parts, we have for $C>\sqrt{2\tau}$,
\begin{align*}
 \E\left[ \e^{Z^{2}/C^{2}} \right] &=1+ \frac{2}{C^{2}}\izf z \e^{z^{2}/C^{2}} 
    \P\bigl\{ |Z|>z \bigr\}\\
 &\le 1+ \frac{2}{C^{2}}\izf z \e^{z^{2}/C^{2}} \e^{-z^{2}/2\tau}\\
 &= 1+\frac{2\tau^{2}}{C^{2}-2\tau}.
\end{align*}
If $C=\sqrt{5\tau/2}$ then this bound is 5, proving \eqref{E:orliczsG}.
\end{proof}

The Orlicz norm is a norm, in the sense that the Orlicz norm of an 
arbitrary sum of random variables is no greater than the sum of the Orlicz norms.
For independent sub-Gaussian random variables $X_1,\dots,X_k$
the variance factors are also sub-additive.

\begin{lemma}  \label{L:Orlicztau}
For any independent centered sub-Gaussian random variables $X_{1},\dots,X_{k}$,
\begin{equation} \label{E:sGsum}
\tau\Bigl(\sum X_{i}\Bigr) \le \sum \tau(X_{i}) ,
\end{equation}
and
\begin{equation} \label{E:sGsum2}
\P\Bigl\{\Bigl|\sum X_{i}\Bigr|>x\Bigr\} \le \exp\Bigl\{-\left(2\sum \tau(X_{i}) \Bigr)^{-1} x^{2} \right\}.
\end{equation}
Also
\begin{equation} \label{E:Orlicztau}
\| X_1+\cdots+X_k\|_\Psi \le \sqrt{5/2}\left( \sum \tau(X_{i}) \right)^{1/2}.
\end{equation}

If $\max\tau(X_{i})\le \tau$ then
\begin{equation} \label{E:sGsum3}
\P\Bigl\{\Bigl|\sum X_{i}\Bigr|>x\Bigr\} \le \exp\left\{- \frac{x^{2}}{2k\tau} \right\}.
\end{equation}
and
\begin{equation} \label{E:Orlicztau2}
\| X_1+\cdots+X_k\|_\Psi \le \sqrt{\frac{5k}{2} \tau }.
\end{equation}
\end{lemma}

\begin{proof}
Statement \eqref{E:sGsum} is Lemma 1.7 of \cite{BK00}, and \eqref{E:sGsum2}
follows by \eqref{E:taubound}. The remainder follows by Lemma \ref{L:sGs}.
\end{proof}

\subsection{Some basic assumptions}
For the upper bounds on growth rate (see Theorem \ref{T:diffrate1}) we will be assuming 
$$
\tX_{t}^{(j)}:=X_{t}^{(j)} - X_{t}^{(0)},
$$ 
with mean $-\tmu^{(i)}$, are sub-Gaussian, writing $\ttau^{(j)}:=\tau(\tX_{t}^{(j)}+\tmu^{(i)})$. We will
also denote the variance factor of $X_{t}^{(0)}$ itself by $\tau^{(0)}$. We assume all of these are finite.
We also define the corresponding subvariances $\ttau^{(j)*}$ and $\tau^{(0)*}$.
Thus
\begin{equation} \label{E:tailUB}
\log\E\left[\e^{\lambda \tX^{(j)}_{t}}\right] \le \frac{\lambda^{2}\ttau^{(j)}}{2}.
\end{equation}
We write 
\begin{align*}
\mu_{A}&:=\max_{(i\to j)\in \M} \E[\log_{+}  A_{t}(j,i)-X_{t}^{(0)}],\\ 
\ttau_{A}&:=\max_{0\le i,j\le d-1} \tau(\log_{+} A_{t}(i,j)-X_{t}^{(0)}).
\end{align*}
We use $\log_{+} A_{t}(i,j):=0\vee \log A_{t}(i,j)$, because the negative part may
actually be infinite. We will use these expressions only for defining an upper bound on the
growth increase induced in a migration step by the matrix elements $A_{t}(i,j)$,
for which increasing $A_{t}(i,j)$ arbitrarily will only improve the bound.

\subsection{Main results} \label{sec:mainresult}
\subsubsection{Same average growth rate at two distinct sites}
\begin{Thm} \label{T:samerate1}
The modulus of continuity of $a$  at $\epsilon=0$ is bounded by $1/\log \epsilon^{-1}$. That is,
\begin{equation} \label{E:logepsMoC0}
\limsup_{\epsilon\downarrow 0} (\log \epsilon^{-1}) \Bigl( a(\epsilon)-a(0) \Bigr)
<\infty.
\end{equation}
If there exists a site $j$ such that $\min\tmu^{(j)}=0$ and $\tX^{(j)}_{t}$ is
not almost surely zero, then
$a$ has modulus of continuity $1/\log \epsilon^{-1}$ at $\epsilon=0$. That is,
\begin{equation} \label{E:logepsMoC}
0<\liminf_{\epsilon\downarrow 0} (\log \epsilon^{-1}) \Bigl( a(\epsilon)-a(0) \Bigr)
\leq\limsup_{\epsilon\downarrow 0} (\log \epsilon^{-1}) \Bigl( a(\epsilon)-a(0) \Bigr)
<\infty.
\end{equation}
\end{Thm}

Notice that this is a fairly generic result, as the lower bound does not depend
on any assumptions about the tails. The upper bound does depend on the sub-Gaussian
assumption for the logarithms of the matrix entries, meaning that heavy-tailed
distributions --- including, but not exclusively, those that are sub-exponential \cite{jT75},
so {\em a fortiori} entries with polynomial order tail behaviour --- could have an even
slower convergence to 0 as $\epsilon$ approaches 0.

\subsubsection{Rare diapause}
\begin{Cor} \label{C:diapause}
In the diapause setting with $M_{t}$ being not deterministic --- that is, at least one entry
has nonzero variance ---
$a$ has modulus of continuity $1/\log \epsilon^{-1}$ at $\epsilon=0$. That is,
\begin{equation} \label{E:logepsMoCcor}
0<\liminf_{\epsilon\downarrow 0} (\log \epsilon^{-1}) \Bigl( a(\epsilon)-a(0) \Bigr)
\leq\limsup_{\epsilon\downarrow 0} (\log \epsilon^{-1}) \Bigl( a(\epsilon)-a(0) \Bigr)
<\infty.
\end{equation}
\end{Cor}

\subsubsection{Distinct growth rates}
For the rest of this section we assume that only site 0 has the maximum
growth rate. Then $a$ still increases for small values of $\epsilon>0$ when
$\Delta_{0}= 0$, but the growth
is only like a power of $\epsilon$. Whether this remains true for $\Delta_{0}<0$
will depend on the exact power,
since the change on the diagonal may produce a decrement linear in $\epsilon$,
so that $a(\epsilon)$ will only be increasing in $\epsilon$ near 0 if the growth
in $a(\epsilon)$ is at a rate faster than linear --- that is, if the power is smaller than 1.

We do not know exactly which power it is, but Theorem \ref{T:diffrate1} gives
upper and lower bounds. The bounds are monotonically increasing in $\mu^{(j)}$
(that is, the larger the gap between the best-site growth rate and the next best,
the smaller the benefit of migration); and decreasing in the variability of the
difference between entries $\tX_{t}^{(j)}$.
However, for this purpose the ``best site'' is determined by a combination of mean difference
and variance, and the distance from the best site, in a way that will be defined precisely
below. 

For the upper bound on growth (that is, the lower bound on the power of $\epsilon$) we
will need the sub-Gaussian variance factor (defined in section \ref{sec:Orlicz}), which
provides an upper bound on the tails, hence on the variability of the differences.
For the lower bound we need the sub-variance $\ttau_{*}^{(j)}$ of $\tX^{(j)}_t$,
defined in \eqref{E:subvar}. We do not need to assume that the sub-variance is positive,
but when all sub-variances are 0 the power of $\epsilon$ in the lower bound is
$\infty$, hence the lower bound is trivial. When $\ttau_{*}^{(j)}$ is positive,
this implies that the Fenchel-Legendre transform
$$
I^{(j)}(x):=\sup_{\lambda>0} \lambda x-\log\E\left[\e^{\lambda \tX_{t}^{(j)}}\right]
$$
is bounded by
\begin{equation} \label{E:FLbound}
I^{(j)}(x)\le \frac{ x^{2}}{2\ttau_{*}^{(j)}}.
\end{equation}
In the special case where $\tX_{t}^{(j)}$ are normally distributed, the
subvariance and the lower standard will both be equal to the
standard deviation. In other cases the ``best'' sites determining the two bounds may be distinct.

We write $\bX$ for the complete collection of all $X^{(j)}_{t}$ for $j=0,1,\dots,d-1$,
$1\le t<\infty$.

If $d=2$ we have an upper bound that $a(\epsilon)-a(0)$ is smaller than
$\epsilon^{4\tmu^{(1)}/(2\tmu^{(1)}+\ttau^{(1)})}$, and lower bound $\epsilon^{4\tmu^{(1)}/\ttau_{*}^{(1)}}$.
For $d>2$ this becomes slightly more complicated for two reasons: First, the
growth will be dominated by one dimension that has the fastest growth; second,
the increment to growth will be smaller if direct transition between the best two
sites is impossible. For this purpose, for each $1\le j\le d-1$ we define $\kappa_{j}$
to be the smallest length of a cycle in $\M$ that starts and ends at 0, and passes through $j$.
(Thus $\kappa_{j}\ge 2$, and is equal to 2 when $A(0,j)>0$ 
and $A(j,0)>0$ both with positive probability.) Define also
\begin{equation} \label{E:rho}
\rho^{(j)}:= \frac{\tmu^{(j)}}{\ttau^{(j)}},\qquad \rho^{(j)}_{*}:= \frac{\tmu^{(j)}}{\ttau_{*}^{(j)}}
\end{equation}

\begin{Thm} \label{T:diffrate1}
Assume all $\tmu^{(j)}>0$, and conditions \eqref{E:tailUB} and \eqref{E:FLbound} hold.
Let $j$ be the site that minimises $\kappa_{j}\rho_{*}^{(j)}$, and
$j'$ the site that minimises $\kappa_{j'}\rho^{(j')}/(1+2\rho^{(j')})$.
If $\Delta_{0}=0$ then for any $\constcp>2\kappa_{j'}$ there are positive constants $C,C'$ (depending
on the $\kappa, \rho,\rho_{*},d, \mu_{A},\ttau_{A}$) such that
for all $\epsilon>0$ sufficiently small,
\begin{equation} \label{E:polyepsMoCd3}
C \epsilon^{2\kappa_{j}\rho^{(j)}_{*}} \le
a(\epsilon)-a(0)\le
C' \epsilon^{2\kappa_{j'}\rho^{(j')}/(1+2\rho^{(j')})} (\log \epsilon^{-1})^{\constcp}.
\end{equation}

Suppose now $\Delta_{0}< 0$. Then
\begin{itemize}
\item If $2\kappa_{j}\rho^{(j)}_{*}<1$ then both
bounds in \eqref{E:polyepsMoCd3} still hold;
\item If 
$
\frac{2\kappa_{j'}\rho^{(j')}}{1+2\rho^{(j')}}<1
   \le  2\kappa_{j}\rho^{(j)}_{*}
$ 
then the upper bound in \eqref{E:polyepsMoCd3} holds;
\item If $\frac{2\kappa_{j'}\rho^{(j')}}{1+2\rho^{(j')}}>1$ then
$a$ is differentiable at $0$, with $a'(0)=-\Delta_{0}$.
\end{itemize}
\end{Thm}

Note that $\rho_{*}^{(j)}=\infty$ when the distribution of $\xi_{t}^{(j)}$
is not heavy-tailed --- for example, very natural choices such as gamma-distributed diagonal elements --- 
making the lower bound on the left-hand side vacuous,
but it remains an open question whether zero subvariance implies that
the approach of $a(\epsilon)$ to 0 is faster than polynomial in $\epsilon$.

\section{Excursion decompositions}   \label{sec:Excursions}
Since we are assuming the unique maximum average growth rate is at site 0,
the maximum growth for the perturbed process will arise from rare excursions away from 0;
in particular, from those that include the (not necessarily unique) site that minimises $\rho\kappa$ in \eqref{E:rho}.

\newcommand{\mE}{\mathcal{E}}
\nc{\hE}{\hat{\mE}}
\nc{\ee}{\mathbf{e}}
\nc{\he}{\hat{\ee}}
\nc{\Kappa}{K}
Define $\mE$ to be the set --- called the {\em excursions from 0} --- 
of cycles in the migration graph that start and end at 0, with no intervening returns to 0.
For an excursion $\ee$ we write $|\ee|$ for the length of the cycle minus 2 --- that is,
the number of time steps spent away from 0.

For a given excursion $\ee$ we define
\begin{align*}
\Kappa(\ee)&:= \bigl\{ 0\le t\le |\ee|+1 \, : \, \ee_{t}\ne \ee_{t+1} \bigr\} \\
\kappa(\ee)&:=\max\bigl\{ \kappa_{j}\, :\, j\in \ee\bigr\};\\
\rho(\ee)&:=\min\left\{\rho^{(j)}\, :\, j\in \ee\right\}.
\end{align*}
Note that 0 and $T$ are always in $\Kappa(\ee)$, and 
the definition of $\kappa^{(j)}$ implies that $\kappa(\ee)\ge \#\Kappa(\ee)$.
We will refer to $\kappa(\ee)$ as the {\em diameter} of $\ee$.

We write $\hE_{T}$ for the collection of sequences of excursions that can be fit into
time $\{1,\dots,T\}$. That is, an element $\he\in\hE_{T}$ has an {\em excursion count} $k(\he)$, such that
each $i\in \{1,\dots,k(\he)\}$ there is a pair $(t_{i},\he_{i})$ with $t_{i}\in \{2,\dots,T-1\}$ and 
$\he_{i}\in \mE$ satisfying
\begin{align*}
t_{i}+|\he_{i}|&< t_{i+1},\\
t_{k(\he)}+|\he_{k(\he)}|&\le T.
\end{align*}
We write the total length of an excursion sequence as 
$$
\|\he\|:=\sum_{i=1}^{k(\he)} |\he_{i}|.
$$
We also write $\hE_{T;k,n,m}$ for the subset of $\hE_{T}$ comprising excursion sequences
whose excursion count is $k$, whose total length is $n$, and the sum of whose diameters
$\kappa(\he_{i})$ is $m$. The {\em null excursion sequence} is the element of $\hE_{T}$ with
$k(\he)=\|\he\|=0$.

The $(0,0)$ entry of the product $R_{T}$ will be a sum of terms that are enumerated by elements of
$\hE_{T}$, corresponding to paths through the sites. We define new random variables
as a function of the realizations of $\bX$ and of $\bA$ (the collection of all matrices $A$)
\begin{equation} \label{E:alphas}
\alpha_{t}(i,j):=
\begin{cases}
 \log\epsilon+\log A_{t}(j,i)-\log X_{t}^{(0)}-\epsilon \Delta_{0}&\text{ if } i\ne j,\\
0&\text{ if } i=j.
\end{cases}
\end{equation}
Given an excursion $\ee$ and a starting time $t_{0}\in \{2,\dots,T-|\ee|\}$ we define the random variables
\begin{equation} \label{E:eebracket}
\ee[t_{0};\bX,\bA]:=\sum_{t\in \Kappa(\ee)} \alpha_{t+t_{0}}(\ee_{t},\ee_{t+1})
   +\sum_{t\in \{1,\dots,|\ee|\}\setminus\Kappa(\ee)} \left(\tX_{t+t_{0}}^{(\ee_{t})}-\epsilon \tD_{\ee_{t}}\right)\, .
\end{equation}
Of course, this sum may be $-\infty$, if it includes a transition at which the corresponding entry of $A$ is 0. But the assumptions imply that it
is finite with nonzero probability if $\ee\in\mE$. Given an excursion sequence $\he=\bigl( (t_{i},\he_{i}) \bigr)_{i=1}^{k}\in \hE_{T}$, we define
\begin{equation} \label{E:hesum}
\he[\bX,\bA]:=\sum_{i=1}^{k} \he_{i}[t_{i};\bX,\bA].
\end{equation}

The quantity we are trying to approximate is
\begin{equation} \label{E:aepsilon0}
a(\epsilon)-a(0)=\lim_{T\to\infty} T^{-1}\left(\log R_{T}(\epsilon)_{0,0}-\sum_{i=1}^{T}X_{i}^{(0)}
\right).
\end{equation}

\begin{lemma}\label{L:allterms}
\begin{equation} \label{E:allterms}
\log R_{T}(0,0)=\sum_{t=1}^{T}X_{t}^{(0)}+\epsilon T \Delta_{0}
  +\log \Bigl(1+  \sum_{\he\in \hE_{T}\setminus\{\he^{0}\}} \e^{ \he[\bX,\bA] }\Bigr\},
\end{equation}
where $\he^{0}$ is the null excursion sequence.
\end{lemma}

\begin{proof}
We have, by definition,
$$
R_{T}(0,0)=\sum_{(x_{0},\dots,x_{T})} \prod_{t=1}^{T} D_{t}^{*}(x_{t},x_{t-1}),
$$
where the summation is over $(x_{0},\dots,x_{T})\in \{0,\dots,d-1\}^{T+1}$ with
$x_{0}=x_{T}=0$. Note that we may restrict the summation to $(T+1)$-tuples such
that $D_{t}^{*}(x_{t},x_{t-1})>0$, which will only be true when $(x_{t-1},x_{t})$ is an edge of $\M$.
Such sequences of states map one-to-one onto excursion sequences. The product corresponding to excursion sequence 
$\he=\bigl( (t_{i},\he_{i})\bigr)_{i=1}^{k}$ is
\begin{equation} \label{E:prodhe}
\prod_{i=1}^{k+1}\Bigl(\prod_{t=t_{i-1}+1}^{t_{i}} D_{t}^{*}(0,0) \Bigr) 
   \cdot \prod_{i=1}^{k}\Bigl( \prod_{t=1}^{|\he_{i}|} D^{*}_{t+t_{i}} \bigl((\he_{i})_{t}, (\he_{i})_{t+1}\bigr) \Bigr),
\end{equation}
where $t_{0}=0$ and $t_{k+1}=T$.

We have $D_{t}^{*}(0,0)=\e^{X_{t}^{0)}+\epsilon\Delta_{0}}$. Thus, we may write the log of the expression in \eqref{E:prodhe} as
\begin{equation} \label{E:prodhe2}
\sum_{t=1}^{T}  \bigl( X_{t}^{0)}+\epsilon\Delta_{0} \bigr)  - \sum_{i=1}^{k} \sum_{t=1}^{|\he_{i}|} 
   \log \frac{D^{*}_{t+t_{i}} \bigl((\he_{i})_{t}, (\he_{i})_{t+1}\bigr)}{D^{*}_{t+t_{i}}(0,0)}
\end{equation}
Noting that
$$
\log \frac{D^{*}_{t+t_{i}} (j,j)}{D^{*}_{t+t_{i}}(0,0)}=
     \tX_{t+t_{i}}^{(j)}-\epsilon \tD_{j},
$$
and for $j\ne j'$,
$$
\log \frac{D^{*}_{t+t_{i}} (j,j')}{D^{*}_{t+t_{i}}(0,0)}=
     \log\epsilon A_{t}(j,j')-\log X_{t}^{(0)}-\epsilon \Delta_{0}=\alpha_{t+t_{i}}(j',j).
$$
Since $\Kappa(\he_{i})$ is precisely the set of $t$ such that $(\he_{i})_{t}\ne (\he_{i})_{t+1}$,
this means that \eqref{E:prodhe2} is precisely the same as $\he_{i}[t_{i};\bX,\bA]$,
which completes the proof.
\end{proof}

Thus
\begin{equation}  \label{E:firstLB}
\log R_{T}(0,0)-\epsilon T\Delta_{0}- \sum_{i=1}^{T} X_{i}^{(0)} \ge  \max_{\he\in\hE_{T}} \he[\bX,\bA]  ,
\end{equation}
and
\begin{equation}  \label{E:firstUB}
\begin{split}
\log R_{T}(0,0)&-\epsilon T\Delta_{0}- \sum_{i=1}^{T} X_{i}^{(0)}\\
  &\le 3\log T+\max_{1\le k,n,m\le T} \Bigl( \log \#\hE_{T;k,n,m} +\max_{\he\in\hE_{T;k,n,m}} \he[\bX,\bA]  \Bigr) .
\end{split}
\end{equation}
Combining this with \eqref{E:aepsilon0} yields the bounds we will use:
\begin{equation}  \label{E:secondLB}
a(\epsilon)-a(0) \ge  \epsilon \Delta_{0}+\liminf_{T\to\infty}T^{-1}\max_{\he\in\hE_{T}} \he[\bX,\bA]  ,
\end{equation}
and
\begin{equation}  \label{E:secondUB}
\begin{split}
a(\epsilon)&-a(0) \le \epsilon \Delta_{0}\\
&+ \limsup_{T\to\infty}T^{-1} \max_{1\le k,n,m\le T}\Bigl( \log  \#\hE_{T;k,n,m} +\max_{\he\in\hE_{T;k,n,m}} \he[\bX,\bA ]  \Bigr) .
\end{split}
\end{equation}

\section{Distinct growth rates}

\subsection{Derivation of the upper bound}  \label{sec:distinctUB}
We prove the upper bound in \eqref{E:polyepsMoCd3}.
We may assume that $\Delta\equiv 0$, since decreasing $\Delta$ can only decrease $a(\epsilon)-a(0)$. Similarly, we may replace $A_{t}(i,j)$ by
$A_{t}(i,j)\vee 1$ for any $(i,j)\in\M$.
That is, we put a floor under those off-diagonal elements which are
allowable migrations. This avoids the nuisance of having entries
be sometimes 0, and again, an upper bound that holds under these conditions
will hold {\em a fortiori} under the original conditions.

We also begin by assuming that there is a single site --- which we may
call site 1 without loss of generality --- that minimizes both $\kappa_{j}$ and
$\rho_{j}$. Since we are only concerned with the behavior of $a(\epsilon)$ as
$\epsilon\downarrow 0$, we will, without loss of generality, confine our attention to
\begin{equation} \label{E:smalleps}
0\le\epsilon\le \e^{-4\ttau_{A} (\rho_{1}\vee 1)}.
\end{equation}

An element of $\hE_{T;k,n,m}$ may be determined by the following choices:
\begin{enumerate}
\item Choose $k$ points out of $T$ where the excursions begin,
yielding no more than $\binom{T}{k}$ possibilities;
\item Choose $k$ numbers for the lengths of the excursions that add up
to $n$, yielding no more than $\binom{n}{k}$ possibilities;
\item Choose $m-2k$ timepoints within these excursions as times when there
is a change of site, yielding at most $\binom{n}{m-2k}$ possibilities;
\item There are no more than $d^{m}$ ways to choose the sites to which
the excursions move at the $m$ times when there is a change.
\end{enumerate}
A crude bound based on Stirling's Formula is
$$
\log\binom{a}{b}\le b+b\log\frac{a}{b},
$$
which holds for all positive integers $b$ and $0\le a\le b$, 
as long as we adopt the convention $0\cdot \log 0=0\cdot \log \infty=0$.
Then
\begin{equation} \label{E:boundlogFTkn}
\log \#\hE_{T;k,n,m} \le m\log d+(m-2k)\log\frac{n}{m-2k}
   +k\log\frac{n}{k} +T\log\frac{T}{k}.
\end{equation}

\begin{Claim}  \label{C:smallP}
For any positive $\constcp>8\kappa_{1}$, and any
$$z\ge \epsilon^{2\kappa_{1}\rho_{1}/(1+2\rho_{1})}
   \cdot (\log \epsilon^{-1})^{\constcp},
$$
we have
\begin{equation} \label{E:theclaim}
\limsup_{T\to\infty}  T^{-1} \log\P\bigl\{ \max_{\he\in\hE_{T;k,n,m}}\bigl( \he[\bX,\bA]+\log\#\hE_{T;k,n,m}\bigr)\ge zT\bigr\} <0
\end{equation}
for all $\epsilon>0$ sufficiently small.
\end{Claim}

We prove this claim in section \ref{sec:proveClaim}, and proceed here under this assumption.
This means that
$$
\sum_{T=T_{0}}^{\infty} \P\left\{ T^{-1}\max_{1\le k,n,m\le T}\Bigl(\log\#\hE_{T;k,n,m}+\max_{\he\in\hE_{T;k,n,m}} \he[\bX,\bA] \Bigr)
   \ge z \right\}<\infty.
$$
By the Borel-Cantelli Lemma, this implies that with probability 1 this event
occurs only finitely often. It follows that the limsup is smaller than $z$ almost surely, and hence,
by \eqref{E:secondUB}, that
\begin{equation} \label{E:aepsdone}
a(\epsilon)-a(0)\le  (\log\epsilon^{-1})^{\constcp} \epsilon^{2\kappa_{1}\rho_{1}/(1+2\rho_{1})}.
\end{equation}

It remains only to clear away the assumption that that $\kappa_{j}$ and $\rho_{j}$ are both minimized at site 1.
We do this by stratifying the excursions further by their diameter (recall the definition from section \ref{sec:Excursions}).
Define 
$$
\breve{\rho}(\kappa):=\min\bigl\{ \rho_{j}\, : \, \kappa_{j}\le \kappa\bigr\}.
$$
If $\ee$ is an excursion with diameter $\kappa$, then any site $j$ included in $\ee$
has $\kappa_{j}\le \kappa$, hence also $\rho_{j}\ge \breve{\rho}(\kappa)$. Furthermore,
\begin{equation} \label{E:minjkappa}
 \min_{1\le j\le d-1}2\rho_{j}\kappa_{j}/(1+2\rho_{j})=\min_{2\le \kappa\le d-1} 2\rho(\kappa)\kappa/(1+2\rho(\kappa))
\end{equation}

The maximum in \eqref{E:secondUB} may be written as a maximum over $(k_{2},\dots,k_{d-1})$, representing
the number of excursions whose diameter is $2,3,\dots,d-1$, with
the constraint $\sum k_{\kappa}=k$. We write $\hE_{T;k,n,m}^{(\kappa)}$ for the excursion sequences consisting of 
$k$ excursions, all of which have diameter $\kappa$; and $\hE_{T;(k_{\kappa}),n,m}$ for the set of excursion sequences that have
exactly $k_{\kappa}$ excursions with diameter $\kappa$. Then $\hE_{T;(k_{\kappa}),n,m}$ naturally includes the direct sum
of $\hE_{T;k_{\kappa},n,m}^{(\kappa)}$.

Using the general fact that the maximum of a sum is smaller than the sum of the maxima,
\begin{align*}
\max_{\he\in\hE_{T;k,n,m}} \he[\bX,\bA]&=\max_{\sum k_{\kappa}=k} \:\max_{\he\in\hE_{T;(k_{\kappa}),n,m}} \he[\bX,\bA]\\
&\le \max_{\sum k_{\kappa}=k} \:\max_{\he_{(\kappa)}\in\hE_{T;k_{\kappa},n,m}^{(\kappa)}} \sum_{\kappa}\he_{(\kappa)}[\bX,\bA]\\
&\le \sum_{\kappa}\max_{1\le  k_{\kappa}\le T} \:\max_{\he_{(\kappa)}\in\hE_{T;k_{\kappa},n,m}^{(j)}} \he_{(\kappa)}[\bX,\bA].
\end{align*}
Thus we have
\begin{align*}
\max_{1\le k,n,m\le T} \Bigl(\log\#\hE_{T;k,n,m}+\max_{\he\in\hE_{T;k,n,m}} \he[\bX,\bA] \Bigr)
&\le \sum_{\kappa}\max_{1\le  k_{\kappa},n_{\kappa},m_{\kappa}\le T} \:\max_{\he_{(\kappa)}\in\hE_{T;k_{\kappa},n_{\kappa},m_{\kappa}}^{(\kappa)}} \he_{(\kappa)}[\bX,\bA].
\end{align*}

Because all excursions in $\hE_{T;k_{\kappa},n_{\kappa},m_{\kappa}}^{(\kappa)}$ pass through only sites $j$
with $\rho_{j}\ge \breve{\rho}(\kappa)$, the same argument used for the upper bound in \eqref{E:aepsdone} may be applied to show that
almost surely 
$$
\limsup_{T\to\infty} T^{-1}\max_{1\le  k_{\kappa},n_{\kappa},m_{\kappa}\le T} \:\max_{\he_{(\kappa)}\in\hE_{T;k_{\kappa},n_{\kappa},m_{\kappa}}^{(\kappa)}} \he_{(\kappa)}[\bX,\bA]
   \le c_{\kappa}(\log\epsilon^{-1})^{c'_{\kappa}} \epsilon^{2\kappa\breve{\rho}(\kappa)/(1+2\breve{\rho}(\kappa)}.
$$
It follows that for $c:=(d-2)\cdot \max c_{\kappa}$ and $c':=\max c'_{\kappa}$,
\begin{align*}
\limsup_{T\to\infty} T^{-1}\max_{1\le k,n,m\le T} \Bigl(\log\#\hE_{T;k,n,m}+&\max_{\he\in\hE_{T;k,n,m}} \he[\bX,\bA] \Bigr)\\
   &\le \sum_{\kappa} c_{\kappa}(\log\epsilon^{-1})^{c'_{\kappa}} \epsilon^{2\kappa\breve{\rho}(\kappa)/(1+2\breve{\rho}(\kappa))}\\
   &\le c(\log\epsilon^{-1})^{c'} \epsilon^{\min_{1\le j\le d-1}2\kappa_{j}\rho_{j}/(1+2\rho_{j})}
\end{align*}
by \eqref{E:minjkappa}, which completes the proof.

\subsection{Derivation of the lower bound}  \label{sec:distinctLB}
We may assume that $\kappa_{j}\rho^{(j)}_{*}$ attains its minimum at $j=1$.
We will write $\kappa$ and $\rho_{*}$ (with no superscript attached) for $\kappa_{1}$
and $\rho_{*}^{(1)}$.
Let $0=j_{0},j_{1},j_{2},\dots,j_{I}=1,j_{I+1},\dots, j_{\kappa-1},j_{\kappa}=0$,
be a cycle from 0 in $\M$, passing through 1.
We may fix a real number $A_{*}$ and $p>0$ such that 
\begin{equation} \label{E:prodA}
\P\Bigl\{ \sum_{i=0}^{\kappa-1}\bigl(\log X_{t}^{(0)}-\log A_{t+i}(j_{i},j_{i+1}) \bigr)<\kappa A_{*}\Bigr\} \ge p.
\end{equation}

Assume that $\epsilon\le \e^{-1}$ and $T>\log\epsilon^{-1}$.
Defining $k=\lfloor T/m\rfloor$ and $m=\lfloor \kappa\log\epsilon^{-1}/\tmu^{(1)}\rfloor$, we will apply \eqref{E:secondLB} by considering only excursions of length exactly $m-1$,
which proceed exactly through the sequence of sites
$0=j_{0},j_{1},\dots,j_{I-1},1,\dots,1,j_{I+1},\dots,j_{\kappa-1},j_{\kappa}=0$, where site 1 is repeated exactly $m-\kappa+1$ times. 
The basic idea is that the excursion fills a time block of length $m$, proceeding
as quickly as possible from 0 to 1, remaining as long as possible at 1, and then
returning to 0.

We define the standard excursion $\ee_{\circ}:=(j_{1},\dots,j_{I-1},1,\dots,1,
j_{I+1},\dots,j_{\kappa-1})$, with $m-\kappa+1$ repetitions of site 1; and an excursion sequence
$\he_{\circ}$ consisting of those pairs $(\ell m+1,\ee_{\ell})$ for which
\begin{equation} \label{E:fgood}
Y_{\ell}:=\ee_{\circ}\bigl[ \ell m+1; \bX,\bA \bigr] >0.
\end{equation}
That is, $\he_{\circ}$ is put together from identical excursions of form $\ee_{\circ}$
which can start only at times $\ell m+1$. Each one of the $k$
possible excursions is included precisely when its contribution
to the sum would be positive.

We have
$$
\he[\bX,\bA]=\sum_{\ell=0}^{k-1}(Y_{\ell})_{+} \,.
$$
Since the excursion contributions $Y_{\ell}$ are all independent, combining
\eqref{E:secondLB} with the Strong Law of Large Numbers yields
\begin{equation} \label{E:helim}
a(\epsilon)-a(0) \ge \frac{\E[(Y_{\ell})_{+}]}{m} -\epsilon\Delta_{0}.
\end{equation}

We now observe that for any $\ell$
$$
Y_{\ell}
=\sum_{i=0}^{I-1} \alpha_{\ell m+i}(j_{i},j_{i+1})
  +\sum_{i=I}^{\kappa-1} \alpha_{(\ell+1) m-\kappa+i}(j_{i},j_{i+1})
   \,+\, \sum_{t=\ell m+I}^{(\ell+1)m-\kappa+I-1} \tX_{t}^{(1)} .
$$
Note that the $\alpha_{t}$ terms are independent of the $\tX_{t}^{(1)}$ terms.

By \eqref{E:prodA}, for any $y>0$,
$$
\P\bigl\{ Y_{\ell}>y \bigr\} \ge p\P\Bigl\{ Z > y+\kappa\bigl(\log \epsilon^{-1}+A_{*}-\epsilon \Delta_{0}\bigr) +(m-\kappa+1)\tmu^{(1)}\Bigr\},
$$
where
$$
Z:=\sum_{t=\ell m+I}^{(\ell+1)m-\kappa+I-1} (\tX_{t}^{(1)}+\tmu^{(1)})
$$
is a mean-zero sub-Gaussian random variable with sub-variance
bounded below by $(m-\kappa+1)\tmu^{(1)}/\rho_{*}$.

By Cram\'er's Theorem \cite[Theorem 2.2.3]{DZ09} and the bound \eqref{E:FLbound},
\begin{align*}
\lim_{m\to\infty} (m-\kappa+1)^{-1}\log\P&\Bigl\{ (m-\kappa+1)^{-1}Z > y+
2\tmu^{(1)}+\frac{\kappa}{m-\kappa+1}\left(\log \epsilon^{-1}+A_{*}+\epsilon \Delta_{0}\right) \Bigr\}\\
 &\ge  -I(2\tmu^{(1)}+y )\\
	&\ge -\frac{\rho_{*}(2\tmu^{(1)}+y )^{2}}{2\tmu^{(1)}}.
\end{align*}

We then have for any $\delta>0$, for fixed $y$ and 
$m$ sufficiently large (that is, for $\epsilon$ sufficiently small)
$$
\P\Bigl\{ \frac{Y_{\ell}}{m} > y \Bigr\} \ge
   p\exp\left\{-\frac{m(2\tmu^{(1)}+y )^{2}\rho_{*}}{2\tmu^{(1)}}-m\frac{\tmu^{(1)}\delta}{2\kappa}
   \right\}.
$$
If
$$
\epsilon\le \exp\left\{-\frac{1}{2\sqrt{\mu\delta}}-\frac{\tmu^{(1)}}{\kappa} \right\},
$$
then by the definition $m:=\lfloor \kappa \log \epsilon^{-1}/\tmu^{(1)}\rfloor$,
$$
\frac{1}{8m\tmu^{(1)} \rho_{*}}\le \frac{\delta}{2}\log \epsilon^{-1},
$$
it follows, by Lemma \ref{L:bounds3}, that
\begin{align*}
\E\bigl[ (Y_{\ell})_{+}\bigr] &\ge  \frac{p}{m\rho_{*}}\exp\left\{-\frac{m(2\tmu^{(1)})^{2}\rho_{*}}
{2\tmu^{(1)}}-\frac{\delta}{2}\log \epsilon^{-1} -\frac{1}{8m\tmu^{(1)}\rho_{*}}\right\} \\
   &\ge \frac{p}{m\rho_{*}} \exp\left\{-\left(2\kappa\rho_{*}+\delta \right)\log \epsilon^{-1}\right\}\\
   &\ge \frac{p\kappa}{\mu^{(1)}\rho_{*} \log\epsilon^{-1}}\epsilon^{2\kappa\rho_{*}+\delta}.
\end{align*}
Combining this with \eqref{E:hesum} completes the proof.

\subsection{Proof of Claim \ref{C:smallP}} \label{sec:proveClaim}
Since the probability is decreasing in $z$, it will suffice to show the statement is
true for
\begin{equation} \label{E:choosez}
z = \epsilon^{2\kappa_{1}\rho_{1}/(1+2\rho_{1})}
   \cdot (\log \epsilon^{-1})^{\constcp},
\end{equation}
where $\constcp$ is any constant larger than $2\kappa_{1}$.

We define 
$$
\zeta:=\frac{k}{T},\quad\nu:= \frac{n}{k},\quad\beta:=\frac{m}{k}.
$$
That is, $\zeta$ is the rate of excursions per unit time; $\nu$ is the average
length of excursions; and $\beta$ is the average diameter of excursions.
We have the constraints $1/\zeta\ge\nu\ge 1$ and $1/\zeta\ge \beta\ge \kappa_{1}$
(since $\kappa_{1}$ is the minimum $\kappa_{j}$, hence the minimum 
number of changes in each excursion).
Then the bound \eqref{E:boundlogFTkn} may be written as
\begin{equation} \label{E:boundlogFTkn2}
\begin{split}
\log& \#\hE_{T;k,n,m} \\
&\le \zeta T\left(\beta\log d
   +(\beta-1)\log\nu -(\beta-2)\log(\beta-2)-\log\zeta\right).
\end{split}
\end{equation}

Suppose now we fix some element $\he$ of $\hE_{T;k,n,m}$,
and list all the states of all the excursions in order as $j_{1},\dots,j_{n}$,
we have
$$
\E\bigl[\he[\bX;\bA]\bigr]\le m\mu_{A}+m\log\epsilon-\sum_{i=1}^{n} \tmu^{(j_{i})}
$$
and the random variable $Y:=\he[\bX;\bA]-\E[\he[\bX;\bA]]$
is sub-Gaussian with variance factor bounded by
$$
\tau(Y)\le \sum_{i=1}^{n} \ttau^{(j_{i})}+m\ttau_{A}
$$
by Lemma \ref{L:Orlicztau}.

For any $x,z>0$, by \eqref{E:taubound}
\begin{align*}
\P\bigl\{ \bigl( \he[\bX,\bA]&+x\bigr))\ge zT\bigr\} 
\le \P\Bigl\{ Y \ge \sum_{i=1}^{n}\tmu^{(j_{i})}+ m\mu_{A}+m\log\epsilon^{-1}
   + zT-x\Bigr\}\\
&\le  \exp\biggl\{-\frac12 \left(\sum_{i=1}^{n} \ttau^{(j_{i})}+m\ttau_{A}
\right)^{-1} \left( \sum_{i=1}^{n}\tmu^{(j_{i})} \,+\,
   m\mu_{A}+m\log\epsilon^{-1}+zT -x\right)^{2} \biggr\}.
\end{align*}
We are assuming that $\rho_{j}$ is minimized
at $j=1$, $\ttau^{(j_{i})}\le \tmu^{(j_{i})}/\rho_{1}$, so that
\begin{align*}
\log\P\bigl\{ \bigl( \he[\bX,\bA]&+x\bigr))\ge zT\bigr\} \\
&\le  -\frac12 \left(\frac{1}{\rho_{1}}\sum_{i=1}^{n} \tmu^{(j_{i})}+m\ttau_{A}
\right)^{-1} \left( 
   \sum_{i=1}^{n}\tmu^{(j_{i})}+m\mu_{A}+m\log\epsilon^{-1}+zT -x\right)^{2}.
\end{align*}

Taking $x=\log\#\hE_{T;k,n,m}$ and substituting \eqref{E:boundlogFTkn2},
we get
\begin{equation} \label{E:Plogbound}
\begin{split}
&\log\P\bigl\{ \max_{\he\in\hE_{T;k,n,m}}\bigl( \he[\bX,\bA]+\log\#\hE_{T;k,n,m}\bigr))
   \ge zT\bigr\} \\
&\le \log\#\hE_{T;k,n,m}+
\max_{\he\in\hE_{T;k,n,m}}\log\P\bigl\{ \bigl( \he[\bX,\bA]+\log\#\hE_{T;k,n,m}\bigr))
   \ge zT\bigr\} \\
&\le T \sup_{S\ge \mu_{*}}  \sup_{\beta\ge\kappa_{1}} \sup_{0\le \zeta\le 1}\zeta \biggl[\beta\log dS/\mu_{*} -\log\zeta
   -(\beta-2)\log(\beta-2)\\
&\qquad - \left(2S/\rho_{1}   +2\ttau_{A}\beta
\right)^{-1} \biggl( S +\log\zeta+\frac{z}{\zeta}+(\beta-2)\log(\beta-2)
   +\beta\Bigl(\c{c1}+\log(S\epsilon)^{-1}  \Bigr)\biggr)^{2} \,\biggr],
\end{split}
\end{equation}
where 
\begin{align*}
\mu_{*}&:=\min \mu^{(j)},\qquad\mu^{*}:=\max \mu^{(j)},\\
\Cr{c1}&:=\mu_{A}+\mu_{*}-\log d,\\
S&:=k^{-1}\sum_{i=1}^{n}\tmu^{(j_{i})} \ge \nu\mu_{*}\ge \mu_{*}
\end{align*}
It is important to note that our assumptions ensure that $\mu_{*}>0$.

We need to show that this supremum is negative. Since the limit as $\zeta\downarrow 0$
is $-\infty$, it will suffice to show this for all possible values of the quantity in
square brackets, which we denote by
\begin{equation} \label{E:Theta}
\begin{split}
 \Theta:=\beta\log dS/\mu_{*}&-\log z+\log u -(\beta-2)\log(\beta-2)\\
&- \frac{\left( S -\log u +u +\log z +(\beta-2)\log(\beta-2)
   +\beta\Bigl(\Cr{c1}+\log(S\epsilon)^{-1} \Bigr)\right)^{2}}
{2S/\rho_{1}  +2\ttau_{A}\beta }
\end{split}
\end{equation}
where $u=z/\zeta$. Here we have taken advantage of the fact that $\log\zeta<0$.

We consider now two different regions for $S$:
\begin{enumerate}
\item $S >\log^{2}\epsilon$\text{ and}
 \label{nurange1}
\item $\mu_{*}\le S \le \log^{2}\epsilon.$
\label{nurange3}
\end{enumerate}
We may assume, without loss of generality, that $\epsilon$ is small enough that the bounds
will be in the right order; that is, $\log^{2} \epsilon>\mu_{*}$.

\nc{\hbeta}{\hat\beta}
\noindent
{\bf Case \ref{nurange1}} 
We define $\hbeta:=(\beta-2)/S\ge 0$, and rewrite $\Theta$ as
\begin{align*}
 \Theta=(S\hbeta+2)\log dS/\mu_{*}-\log &z+\log u -S\hbeta\log S\hbeta\\
&- \frac{\left( S -\log u +u +\log z +S\hbeta\log S\hbeta
   +(S\hbeta+2)\Bigl(\Cr{c1}+\log(S\epsilon)^{-1} \Bigr)\right)^{2}}
{2S/\rho_{1}  +2\ttau_{A}S\hbeta +4\ttau_{A}}.
\end{align*}
Again, since we are only interested in showing that this expression must be negative
and bounded away from zero, we may equally well analyze $\Theta/S$. We have then
\begin{equation} \label{E:ThetaS}
\begin{split}
\frac{\Theta}{S}&=\c{c2}\hbeta -\hbeta\log\hbeta  +S^{-1}\left(2\Cr{c2} -\log z+\log u\right)  \\
&\quad- \frac{\left(1 +S^{-1}(2\Cr{c1}+ u-\log u+\log \epsilon^{-2}) +S^{-1}\log z
   +\hbeta(\Cr{c1}+\log\hbeta\epsilon^{-1})+2S^{-1}\log S \right)^{2}}
   {2/\rho_{1}  +2\ttau_{A}\hbeta +4\ttau_{A}/S}\\
\end{split}
\end{equation}
where $\Cr{c2}:=\log d/\mu_{*}$.

\nc{\hbm}{\hbeta_{max}}
We may find $\epsilon_{0}$ such that for all $0<\epsilon\le \epsilon_{0}$, 
\begin{align}
&\log\epsilon<-2|\Cr{c1}|-2,\text{ which implies }\hat\beta\bigl( \Cr{c1}+\log \hbeta\epsilon^{-1}\bigr)  \ge -\epsilon \e^{-\Cr{c1}-1}\ge -\e^{-3} ; \label{E:epsbound2}\\
&\log \epsilon<-\sqrt{2\ttau_{A}\rho_{1}} ;\label{E:epsbound3}\\
&\frac{2|\Cr{c2}|+\log (4/\rho_{1}+2+\ttau_{A})}{\log^{2}\epsilon}
   + \frac{2\rho_{1}\kappa_{1}}{(2\rho_{1}+1)\log\epsilon^{-1}}<\frac14 ; \label{E:epsdelta}\\
&\frac{(1+2\Cr{c10})^{2}}{\Cr{c10}\bigl(\Cr{c1}+\log\epsilon^{-1}\bigr) -2\Cr{c2}}
     <2\Cr{c10}
     ,\text{ where }\c{c10}:=(4/\rho_{1}+2\tau_{A})^{-1}\label{E:epsboundalpha}\\
&\Cr{c10}(\Cr{c1}+\log\epsilon^{-1})-2|\Cr{c2}| \ge 1.\label{E:epsalpha2}
\end{align}

We consider the derivative of $\Theta/S$ with respect to $u$, which is
$$
\frac{1}{Su}\left[1 -\frac{1+\hbeta(\Cr{c1}+\log\hbeta\epsilon^{-1})+S^{-1}(2\Cr{c1}+ u-\log u+\log \epsilon^{-2}) +S^{-1}\log z
   +2S^{-1}\log S}{1/\rho_{1}  +\ttau_{A}\hbeta +2\ttau_{A}/S}(u-1)\right].
$$
Because the numerator and denominator of the large fraction are both positive, $\Theta/S$ must
be increasing for $u\in (0,1)$. For $u>1$, using the bound \eqref{E:epsbound2}, the derivative is bounded above by
$$
\frac{1}{Su}\left[1 -\frac{1-e^{-3}+\hbeta(\Cr{c1}+\log\hbeta\epsilon^{-1}) +S^{-1}\log z}
   {2/\rho_{1}  +\ttau_{A}\hbeta }(u-1)\right],
$$
which will be negative for $u\ge 4/\rho_{1}+2+\ttau_{A}$. Thus
$$
\frac{\Theta}{S}\le \Cr{c2}\hbeta -\hbeta\log\hbeta 
- \frac{\left(\frac12
   +\hbeta(\Cr{c1}+\log\epsilon^{-1}) +\hbeta\log\hbeta \right)^{2}}
   {4/\rho_{1}  +2\ttau_{A}\hbeta },
$$
by \eqref{E:epsdelta}.

If $\hbeta >1$ then
$$
\frac{\Theta}{S}\le -\Cr{c10}(\Cr{c1}+\log\epsilon^{-1})^{2}+|\Cr{c2}|  \le -1
$$
by \eqref{E:epsalpha2}.

If $\hbeta\in [0,1]$, we may bound this by
$$
\frac{\Theta}{S}\le \Cr{c2}\hbeta -\hbeta\log\hbeta 
- \Cr{c10}\left(\frac{1}{2}   +\hbeta(\Cr{c1}+\log\epsilon^{-1}) +\hbeta\log\hbeta \right)^{2}
$$
where $\Cr{c10}$ is defined in \eqref{E:epsdelta}.
Given the bounds $\hbeta\log\hbeta\ge -\sqrt{\hbeta}$ and $\hbeta\le\sqrt{\hbeta}$
for $\hbeta\in[0,1]$, and given that 
$\Cr{c1}+\log\epsilon^{-1}>1$, we obtain the bound
$$
\left(1   +\hbeta(\Cr{c1}+\log\epsilon^{-1}) +\hbeta\log\hbeta \right)^{2}
   > \frac12-2\sqrt{\hbeta}+2(\Cr{c1}+\log\epsilon^{-1})\hbeta,
$$
and so
$$
\frac{\Theta}{S}\le \Bigl( \Cr{c2}-2\Cr{c10}\bigl(\Cr{c1}+\log\epsilon^{-1}\bigr)\Bigr)
\hbeta+(1+2\Cr{c10})\sqrt{\hbeta}   - \frac{\Cr{c10}}{4} 
   \le \frac{(1+2\Cr{c10})^{2}}{8\Cr{c10}\bigl(\Cr{c1}+\log\epsilon^{-1}\bigr) -4\Cr{c2}}
     - \frac{\Cr{c10}}{4} <0
$$
by \eqref{E:epsboundalpha}.\\

\noindent
{\bf Case \ref{nurange3}}
In the remaining cases we simplify the dependence on $\beta$ by using 
$$
(\beta-2)\log(\beta-2)\ge -\e^{-1}\ge -\beta\log 2 \text{ for }\beta\ge 2,
$$
to obtain
\begin{equation} \label{E:NewTheta}
\begin{split}
 \Theta\le\wt\Theta:=\beta\log 2dS&/\mu_{*}-\log z+\log u \\
&- \left(2S/\rho_{1}  +2\ttau_{A}\beta \right)^{-1} 
  \biggl( S -\log u +u +\log z 
   +\beta\Bigl(\Cr{c1}+\log(2S\epsilon)^{-1} \Bigr)\biggr)^{2} .
\end{split}
\end{equation}
We assume $\epsilon\le\epsilon_{0}$, chosen such that for all such $\epsilon$,
\begin{align}
&\frac{\kappa_{1}}{1+2\rho_{1}}\log \epsilon^{-1}+\kappa_{1}\left(\Cr{c1}-\log 2\mu_{*}\epsilon \right)>0;
   \label{E:epsbound23}\\
&\frac{\log\log\epsilon^{-1}}{\log\epsilon^{-1}}
   \le \frac{\rho_{1}}{2\rho_{1}+2};\label{E:epsbound25}\\
   &\log\epsilon^{-1}>\ttau_{A}\rho_{1} -\Cr{c1}+ \log4\log\epsilon^{-1}; 
     \label{E:epsboundC}\\
   &\log\log\epsilon^{-1}>\frac{\Cr{c5}}{(2\rho_{1}+1)(\constcp-2\kappa_{1})}  \label{E:epsbound26}\\
  &\quad \text{where } \label{E:c5}
\Cr{c5}:=\Cr{c4}-2\rho_{1}\kappa_{1}(\Cr{c1}-\log 2)-2\rho_{1}^{2}\ttau_{A} \kappa_{1},\quad
   \Cr{c4}:= \frac{1}{\rho_{1}}+ \kappa_{1}\log 2d/\mu_{*};\\
 &\log\epsilon^{-1}>\max\left\{ 8+2\Cr{c1}\, ,\, 8\Cr{c8}\bigl(\Cr{c2}-\log \Cr{c8}\bigr)\right\}, \\
&\qquad
 \text{where } 
  \c{c8}:=64\left(\frac{1}{\rho_{1}}  +\ttau_{A}e^{|\Cr{c2}|} +\frac{\ttau_{A}}{\mu_{*}}\right).\label{E:c8} 
\end{align}

We have
\begin{align*}
\frac{\partial\wt\Theta}{\partial u}&=
   \frac{1}{u}\left(1-\rho_{1}\Bigl(u-1\Bigr)\frac{S -\log u +u +\log z
   +\beta\bigl(\Cr{c1}-\log 2S -\log\epsilon  \bigr)}
   {S   +\rho_{1}\ttau_{A}\beta}\right)\\
&=    \frac{1}{u}\left(1-\rho_{1}\Bigl(u-1\Bigr)\frac{S -\log u +u +\log z\epsilon^{-2\beta\rho_{1}/(2\rho_{1}+1)}
   +\beta\bigl(\Cr{c1}-\log 2S -(2\rho_{1}+1)^{-1}\log\epsilon  \bigr)}
   {S   +\rho_{1}\ttau_{A}\beta}\right).
\end{align*}
We know that $u-\log u>0 $, and
$$
\log z\epsilon^{-2\kappa_{1}\rho_{1}/(2\rho_{1}+1)}>0 \text{ and}\quad
  \Cr{c1}-\log 2S -(1+2\rho_{1})^{-1}\log\epsilon\ge \rho_{1}\ttau_{A}
$$
by the definition of $z$ \eqref{E:choosez} and the upper bound on $S$ in this case, respectively,
which together imply that
$$
\frac{S -\log u +u +\log z\epsilon^{-2\beta\rho_{1}/(2\rho_{1}+1)}
   +\beta\bigl(\Cr{c1}-\log 2S -(2\rho_{1}+1)^{-1}\log\epsilon  \bigr)}
   {S   +\rho_{1}\ttau_{A}\beta}\ge 1.
$$
It follows that $\wt\Theta$ is increasing in $u$ on $(0,1)$ and
decreasing in $u$ on ${(1+1/\rho_{1},\infty)}$. Consequently,
\begin{equation} \label{E:BigTheta}
\Theta\le\wt{\wt\Theta}:= \beta\log 2dS/\mu_{*} - \log z+\frac{1}{\rho_{1}} - \left(2S/\rho_{1}   +2\ttau_{A}\beta
\right)^{-1} \biggl( S +\log z
   +\beta\Bigl(\Cr{c1}-\log 2S \epsilon  \Bigr)\biggr)^{2},
\end{equation}
where we use the assumption \eqref{E:epsbound23}, which implies that
$S+\log z+\beta\left(\Cr{c1}-\log 2S\epsilon \right)>0$.

The next step is to optimize with respect to $\beta$. We replace $\beta$ by $x=\beta-\kappa_{1}$
(which now takes values on $[0,\infty)$. Then in the notation of Lemma \ref{L:bounds} we set
\begin{align*}
A=\log\frac{2dS}{\mu_{*}},\quad
B= \log\epsilon^{-1}+\Cr{c1}-\log 2S,\quad C=S+\log z+\kappa_{1}B,\quad D=\frac{2}{\rho_{1}}S+2\ttau_{A}\kappa_{1}, \quad E=2\ttau_{A}.
\end{align*}
Observe that these are all positive, and by \eqref{E:epsboundC}
\begin{align*}
C&>S+\frac{\kappa_{1}}{2+4\rho_{1}} \log\epsilon^{-1},\\
 B&>\frac12\log\epsilon^{-1}.
\end{align*}

We then have
\begin{align*}
BD-CE&= \frac{2}{\rho_{1}} SB-2-\ttau_{A}S -2\ttau_{A}\log z\\
&=\frac{2}{\rho_{1}}S \left[  \log\epsilon^{-1}+\Cr{c1}-\log{2S}-\ttau_{A}\rho_{1}\right]
   +\frac{4\ttau_{A} \rho_{1}\kappa_{1}}{1+2\rho_{1}} \log\epsilon^{-1}\\
&>0 \text{ by \eqref{E:epsboundC}};\\
AD&< \left(\frac{2S}{\rho_{1}}+2\ttau_{A}\kappa_{1}\right) \left(\log \frac{2d}{\mu_{*}} +2\log \log\epsilon^{-1}\right)\\
&<S\cdot \frac{\log \epsilon^{-1}}{2}\\
&<BC.
\end{align*}
Thus \eqref{E:CEBD} holds, and we may conclude that 
$\wt{\wt\Theta}$ is decreasing in $x$ on $(0,\infty)$.
In particular, the maximum over the allowable range $[\kappa_{1},\infty)$ is attained at $\beta=\kappa_{1}$.

We now know that the maximum value of $\Theta$ is bounded by
\begin{equation} \label{E:BigTheta2}
\begin{split}
\sup \Theta&\le \sup_{S} \,\kappa_{1}\log S
   -\frac{\Bigl(S +\log z  +\kappa_{1}\Cr{c1} -\kappa_{1} \log 2S\epsilon  \Bigr)^{2}}{2S/\rho_{1}   +2\ttau_{A}\kappa_{1}}  - \log z+\Cr{c4}\\
 &\le 2\kappa_{1} \log \log \epsilon^{-1}
   -\inf_{S}\frac{\Bigl(S  +\Cr{c1} +\frac{\kappa_{1}}{1+2\rho_{1}} \log \epsilon^{-1} -\kappa_{1} \log^{2} \log \epsilon^{-1} \Bigr)^{2}}{2S/\rho_{1}   +2\ttau_{A}\kappa_{1}}  - \log z+\Cr{c4}\\
\end{split}
\end{equation}
where $\c{c4}$ is defined in \eqref{E:c5}.

By Lemma \ref{L:bounds} (applied now with $S$ in the role of $x$, and
$A=0$),
\begin{align*}
\sup \Theta &\le (2\kappa_{1}-\constcp)\log\log\epsilon^{-1} +
    \frac{2\rho_{1}\kappa_{1}}{1+2\rho_{1}} \log \epsilon^{-1}+\Cr{c4}\\
   &\hspace*{2cm} -2\rho_{1}\left(\kappa_{1}(\Cr{c1}-\log 2) +\frac{\kappa_{1}}{1+2\rho_{1}} \log \epsilon^{-1} -(2\kappa_{1}-\constcp) \log \log \epsilon^{-1}\right)-2\rho_{1}^{2}\ttau_{A}\kappa_{1} \\
   &= (2\rho_{1}+1)(2\kappa_{1}-\constcp)\log\log\epsilon^{-1}+\Cr{c5}.
\end{align*}
where $\c{c5}$ is defined in \eqref{E:c5}.
By \eqref{E:epsbound26} this is negative, completing this Case.\\

We conclude by filling in the details that connect $\Theta$ to the quantity we are trying to bound.
Putting both cases together, we see that there is a constant (expressible in terms only of
$\constcp,\rho_{1},\kappa_{1},\mu_{*},\ttau_{A},d$, positive $\epsilon_{0}$ and $r>0$ 
such that for any fixed $\epsilon\in (0,\epsilon_{0})$, and every
$$
z\ge \epsilon^{2\kappa_{1}\rho_{1}/(1+2\rho_{1})} (\log \epsilon^{-1})^{8\kappa_{1}+1}
$$
the supremum of $\Theta$ is less than $-r$. 

Holding $\epsilon\in (0,\epsilon_{0})$ fixed, write 
$$
R(U):=-T \sup_{S\ge \mu_{*}}  \sup_{\beta\ge\kappa_{1}} \sup_{u\ge U}
   \frac{\Theta(S,\beta,u)}{u}.
$$
By the above argument, we know that $R(U)> 0$ for any $U\ge z$.
According to \eqref{E:Plogbound}, we have
\begin{align*}
\lim_{T\to\infty}T^{-1}\log\P\bigl\{ \max_{\he\in\hE_{T;k,n,m}}\bigl( \he[\bX,\bA]+\log\#\hE_{T;k,n,m}\bigr))
   \ge zT\bigr\} 
&\le \sup_{S}\sup_{\hbeta}\sup_{u}\frac{z}{u} \Theta(S,\hbeta,u)\\
&= -z R(z),
\end{align*}
(since $\zeta=z/u$)
so we need to show that $R(z)>0$.
For any $U'\ge U\ge z$,
$$
R(U)\ge \min\left\{ \frac{r}{U'}, R(U') \right\}, \text{ and } 
   R(U')\ge \frac{U}{U'} R(U).
$$
Thus, we will be done if we can show that $R(U)>0$ for any $U$.

Take 
$U=\Cr{c8}\max\{2\Cr{c2},-\log z,\Cr{c8}\}$, where $\Cr{c8}$ is defined in \eqref{E:c8},
and also large enough that . 
Then for any $u\ge U$, and any $\hbeta\ge 0$ and any $S\ge\mu_{*}$, by \eqref{E:ThetaS}
$$
\Theta \le 
  S\left(\Cr{c2}\hbeta -\hbeta\log\hbeta \right)  +2\Cr{c2} -\log z+\log u
   -    2S^{-1}\Cr{c8}^{-1}\left(S(1 +2\hbeta\log\epsilon^{-1})+u\right)^{2},
$$
using the facts that 
\begin{align*}
u-\log u+\log z&>u/4,\\
\hbeta(\Cr{c1}+\log\hbeta\epsilon^{-1})&>-0.02+\frac{\hbeta}{2} \log\epsilon^{-1} \quad\text{by the first part of \eqref{E:c8}, and}\\
1+2S^{-1}\log S&> 0.27.
\end{align*}

Note that the bound \eqref{E:c8} on $\epsilon_{0}$ and the definition of $U$ imply that
for $u\ge U$ and $\epsilon<\epsilon_{0}$
\begin{align*}
\Cr{c2}\hbeta -\hbeta\log\hbeta - \Cr{c8}\left( 1 +8\hbeta \log\epsilon^{-1} \right)&<0,\\
u^{-1}\left(2\Cr{c2} -\log z+\log u\right)<4\Cr{c8}^{-1}.
\end{align*}
We then have for all $u\ge U$,
$$
\frac{\Theta}{u}\le -\Cr{c8}^{-1}\left( \frac{S}{u}+2\frac{u}{S}\right)\le -2\sqrt{2}\Cr{c8}^{-1}<0,
$$
so that $R(U)>0$, completing the proof.

\begin{lemma}  \label{L:bounds}
For real $A,B,C,D,E$, with $D,E>0$ and $B^{2}>AE$, 
the expression ${Ax-{(Bx+C)^{2}/(Ex+D)}}$ is concave over $x\ge -E/D$ and achieves its
maximum over $x$ in this range --- so, including, in particular, all positive $x$ --- at 
\begin{equation} \label{E:x0}
x_{0}=\frac{|CE-BD|-\tB D}{\tB E},
\end{equation}
where $\tB:=\sqrt{B^{2}-AE}$. 
The maximum value attained on this interval is
\begin{equation} \label{E:themaximum1}
\sup_{x\ge 0} Ax-\frac{(Bx+C)^{2}}{Ex+D}=
\frac{1}{E^{2}}\cdot\begin{cases}
-2(B+\tB)\left( CE-BD\right)-ADE&\text{if } CE\ge BD,\\
2(B-\tB)\left( BD-CE\right)-ADE&\text{if } CE<BD.
\end{cases}
\end{equation}

If
\begin{equation} \label{E:CEBD}
\frac{A}{B}<\frac{C}{D}<\frac{B}{E}
\end{equation}
then $x_{0}<0$, and the function is decreasing in $x$ for all $x>0$.3
%
%
\end{lemma}

\begin{proof}
Setting $y=Ex+D$, we may rewrite the function as
\begin{align*}
f(y)&=\frac{A}{E}y-\frac{AD}{E}- y^{-1}\left(\frac{B}{E}y+C-\frac{BD}{E}\right)^{2}\\
&=-\frac{\tB^{2}}{E^{2}}y-\left(\frac{BD}{E}-C\right)^{2}y^{-1}
   -\frac{2BC}{E}+\frac{2B^{2}D}{E^{2}}-\frac{AD}{E}
\end{align*}
This is a concave function that goes to $-\infty$ at $y=0$. If
$\tB^{2}\le 0$ then it is asymptotic to a line with positive or zero
slope as $y\to\infty$, so it is increasing for all $y>0$ (which includes
all $x>0$). If $\tB^{2}>0$ then it goes to $-\infty$ as $y\to\infty$.
The first two derivatives are
\begin{align*}
f'(y)&=- \frac{\tB^{2}}{E^{2}} +\left(\frac{BD}{E}-C\right)^{2}y^{-2},\\
f''(y)&=-2\left(\frac{BD}{E}-C\right)^{2}y^{-3},
\end{align*}
showing that the second derivative is negative for $y>0$, hence the function is concave. 
The first derivative is zero at $y_{0}:=|CE-BD|/\tB$, which is equivalent to \eqref{E:x0}.
The maximum value is
$$
f(y_{0})
   =-\frac{2\tB}{E^{2}}\left| CE-BD\right|
   -\frac{2B}{E^{2}}\left(CE-BD\right)-\frac{AD}{E},
$$
which simplifies directly to \eqref{E:themaximum1}

Suppose now that \eqref{E:CEBD} holds. Since $\tB\ge B-\frac{AE}{B}$, $CE-BD<0$ implies
$$
x_{0}=\frac{BD-CE-\tB D}{\tB E}\le \frac{AD-BC}{B\tB}<0.
$$
\end{proof}

%

\begin{lemma} \label{L:bounds3}
Let $Z$ be a random variable with tail bound 
$$
\P\bigl\{ Z>z\bigr\} \ge \e^{-(z_{0}+z)^{2}/\tau_{Z}} \text{ for all }z>0.
$$
Then 
$$
\E[Z_{+}]\ge \e^{-z_{0}^{2}/\tau_{Z}}\cdot \e^{-\tau_{Z}/4z_{0}^{2}}\cdot \frac{\tau_{Z}}{2ez_{0}}.
$$
\end{lemma}
\begin{proof}
\begin{align*}
\E[Z_{+}]&=\int_{0}^{\infty} \P\bigl\{ Z>z\bigr\} dz\\
&\ge \e^{-z_{0}^{2}/\tau_{Z}} \int_{0}^{\infty} \e^{-z(2z_{0}+z)/\tau_{Z}} dz\\
&\ge z_{*} \e^{-z_{0}^{2}/\tau_{Z}}  \e^{-z_{*}(2z_{0}+z_{*})/\tau_{Z}}
\end{align*}
for any fixed $z_{*}>0$, since the integrand is decreasing. Taking $z_{*}=\tau_{Z}/2z_{0}$ gives us the result.
\end{proof}

\section{Equal rates}
The equal rates case is, on the one hand, easier, because the average ``cost'' of an excursion is zero,
hence there is no need to keep track of the time spent on an excursion.
On the other hand, this also makes the problem more challenging, because there is a much larger
class of paths that need to be considered. When the average rates are unequal, only
paths that spend nearly all their time at 0 can even plausibly contribute to the
upper bound. With equal rates, all paths --- or, at least, all paths that spend nearly
all of their time at the sites with maximal rates --- are essentially on an equal footing.
This makes the combinatorial bounds more challenging.

We will assume throughout this section that $\Delta=0$. The upper bound
clearly can only be improved by making some elements of $\Delta$
positive. The lower bound can only be reduced by $\epsilon \max\Delta_{j}$,
which will be much smaller than any multiple of $(\log\epsilon^{-1})^{-1}$
for $\epsilon$ sufficiently small.

\subsection{Trajectories} \label{sec:Trajectories}
Instead of considering excursions, which emphasize time spent at 0 as the
baseline, we will base our analysis on trajectories, which will simply be
paths in the migration graph $\M$. The set of all trajectories
of length $T$ will be denoted $\F_{T}$. The set of changepoints of a trajectory $f$ will be denoted
$$
K(f):=\{t\, :\, f_{t}\ne f_{t+1}\}.
$$
We write $\F_{T,k}$ for the set of trajectories with exactly $k$ changepoints,
and we have $\binom{T}{k}\le\#\F_{T,k}\le d^{k} \binom{T}{k}$.
We endow $\F_{n}$ with the $L^{2}$ norm $\|\cdot\|_{2}$, defined to
be the square root of the Hamming distance (the number of times at which the trajectories
are not equal). The {\em null trajectory}
$f^{(0)}$ will denote the path that stays at 0 for all $T$ steps.

\nc{\ff}{\mathbf{f}}
We then have the random variables
$$
Z_{f}:=f[\bX,\bA]:=\sum_{t\in K(f)} \log A_{t}(f_{t+1},f_{t})
   +\sum_{t\in \{0,\dots,T-1\}\setminus K(f)} \tX_{t}^{(f_{t})}.
$$
(We will use the $Z_{f}$ notation for brevity when there is no need
to emphasize the dependence on $\bX$ and $\bA$.)

We have the analogue of Lemma \ref{L:allterms}:
\begin{lemma} \label{L:allterms2}
\begin{equation} \label{E:allterms2}
\log R_{T}(0,0)=\sum_{t=1}^{T}X_{t}^{(0)}
  +\log \Bigl(1+  \sum_{f\in \F_{T}\setminus\{f^{(0)}\}} \e^{ f[\bX,\bA] }\epsilon^{K(f)}\Bigr),
\end{equation}
where $f^{(0)}$ is the null trajectory. Thus
\begin{equation} \label{E:allterms3}
\begin{split}
\liminf_{T\to\infty} T^{-1}\max_{f\in \F_{T}} f[\bX,\bA] -K(f)&\log\epsilon^{-1}
\le a(\epsilon)-a(0)\\
&\le  \limsup_{T\to\infty}T^{-1} \Bigl( \log  \#\F_{T} +\max_{f\in\F_{T}} f[\bX,\bA ]  -K(f)\log\epsilon^{-1}\Bigr)
\end{split}
\end{equation}
\end{lemma}

\subsection{Proof of the lower bound}  \label{sec:sameLB}
We begin by assuming $X_{i}^{(0)}$ bounded above by a real number $R$.
We may find paths $0=j_{0},j_{1},\dots,j_{I}=1$
and $1=j'_{0},j'_{1},\dots,j'_{I'}=0$ in $\M$, with $I+I'=\kappa_{1}$. 

Now consider the random variables
\begin{align*}
Z(t_{1},t_{2}):= \max\Bigl\{ \sum_{t=s_{1}+I}^{s_{2}-I'}\tX^{(1)}_{t}
&+\sum_{i=0}^{I-1}\log A_{s_{1}+i} (j_{i+1},j_{i})-X_{s_{1}+i}^{(0)} \\
  &\quad+ \sum_{i=1}^{I'} \log A_{s_{2}-I'+i}(j'_{i+1},j'_{i})-X^{(0)}_{s_{2}-I'+i}\, :\\
&\qquad \qquad\qquad t_{1}\le s_{1}\le s_{2}\le t_{2},\quad s_{2}\ge s_{1}+\kappa_{1}-1\Bigr\}.
\end{align*}
Note that for any fixed $\ell\ge 1$, the collection of random variables 
$$
\bigl\{ Z(\ell j,\ell (j+1)-1)\, :\,j=0,1,\dots, k-1 \bigr\}
$$
are i.i.d., and
\begin{equation}  \label{E:maxZf}
\max\bigl\{ Z_{f}\, :\, f\in \F_{\ell k,\kappa_{1}k} \bigr\} \ge \\
\sum_{j=0}^{k-1} Z\bigl( \ell j,\ell(j+1)-1\bigr) 
\end{equation}

The distribution of $Z(t_{1},t_{2})$
depends only on $t_{2}-t_{1}$. This allows us to define
\begin{equation} \label{E:Gammal}
\Gamma_{\ell}:=\ell^{-1/2}\E\bigl[ Z(\ell_{1},\ell_{1}+\ell-1)\bigr].
\end{equation}
It follows from \eqref{E:maxZf} and the Strong Law of Large Numbers that for fixed $\ell$,
\begin{equation} \label{E:maxZf2}
\lim_{k\to\infty}k^{-1}\max\bigl\{ Z_{f}\, :\, f\in \F_{\ell k,\kappa_{1}k} \bigr\} \ge  \ell^{1/2} \cdot \Gamma_{\ell}
\end{equation}
for any $\ell$. We now fix $\ell=\lceil2\kappa_{1}\log\epsilon^{-1}/\Gamma_{*}\rceil^{2}$, where
$\Gamma_{*}$ is defined in Lemma \ref{L:Cell}.

For any $T\ge\ell$ we then have
\begin{align*}
T^{-1}\max\bigl\{ Z_{f}-K(f) \log\epsilon^{-1}\, :\, &f\in \F_{T} \bigr\} \ge T^{-1}\max\bigl\{ Z_{f}-K(f)\log\epsilon^{-1}\, :\, f\in \F_{k\ell} \bigr\}\\
&\ge T^{-1}\max\bigl\{ Z_{f}\, :\, f\in \F_{k\ell,\kappa_{1}k} \bigr\} -T^{-1}\kappa_{1}k \log\epsilon^{-1}, \\
&\qquad\text{ since }
  K(f)=\kappa_{1}k \text{ for any }f\in \F_{\ell k,\kappa_{1}k}\\
&\ge  \left(1-\frac{\ell}{T}\right) \ell^{-1} k^{-1} \max\bigl\{ Z_{f}\, :\, f\in \F_{k\ell,\kappa_{1}k} \bigr\} -\ell^{-1}\kappa_{1} \log\epsilon^{-1}
\end{align*}
where $k=\lfloor T/\ell\rfloor$.

Hence, by Lemma \ref{L:Cell},
\begin{equation} \label{E:maxZf3}
\begin{split}
\lim_{T\to\infty}T^{-1}\max\bigl\{ Z_{f}-K(f)\log\epsilon^{-1}\, :\, f\in \F_{T} \bigr\} &\ge 
   \ell^{-1/2}\Gamma_{*} -\ell^{-1}\kappa_{1} \log \epsilon^{-1}\\
   &\ge \frac{\Gamma_{*}^{2}}{4\kappa_{1}\log\epsilon^{-1}+2\Gamma_{*}} .
 \end{split}
\end{equation}
Thus, for any $\delta>0$,
$$
a(\epsilon)-a(0)\ge \left(\frac{\Gamma_{*}^{2}}{4\kappa_{1}} -\delta\right) \frac{1}{\log\epsilon^{-1}}
$$
for positive $\epsilon$ sufficiently small.

\begin{lemma}  \label{L:Cell}
$$
\liminf_{\ell\to\infty} \Gamma_{\ell} =\Gamma_{*}\ge \sqrt{8/\pi}\log 2\cdot \sigma > 0,
$$
where $\sigma^{2}=\operatorname{Var}(\tX_{t}^{(1)})$. In particular, for $\ell_{0}$ sufficiently
large, $\inf_{\ell\ge \ell_{0}} \Gamma_{\ell}>0$.
\end{lemma}

\begin{proof}
We consider, for definiteness, $Z(0,\ell-1)$. We suppose first that $X_{t}^{(0)}$ is almost-surely
bounded by some number $R$.

By the definition of $\M$ and the fact that $(A_{t})$ are i.i.d., 
$$
A_{*}:=\operatorname{ess}\inf\E\Bigl[\sum_{i=0}^{I-1}\log A_{i} (j_{i+1},j_{i}) \cond \eu{D} \Bigr] +
\operatorname{ess}\inf\E\Bigl[\sum_{i=0}^{I'-1}\log A_{i} (j'_{i+1},j'_{i}) \cond \eu{D} \Bigr]  >
-\infty.
$$
That is, the expected value of either sum is at least $A_{*}$, conditioned on any realization of
$(D_{t})_{t=1}^{\infty}$. 

Define for $0\le t<1$
$$
W_{t}:= \frac{1}{\sqrt{\ell}\sigma} \sum_{i=1}^{\lfloor t\cdot \ell \rfloor}
  \tX_{i}^{(1)}\, .
$$
Then 
\begin{align*}
Z(0,\ell-1)&=\hspace*{-5mm}\max_{\frac{I}{\ell}\le t\le t'\le 1-\frac{I}{\ell}}\biggl(\sqrt{\ell}\sigma (W_{t'}-W_{t})
   +\sum_{i=0}^{I-1}\log A_{\lfloor t\ell \rfloor +1+i} (j_{i+1},j_{i})-X_{\lfloor t\ell \rfloor +1+i}^{(0)} \\
   &\hspace*{6.5cm}+ \sum_{i=1}^{I'}\log A_{\lfloor t'\ell \rfloor-1-I'+i}(j'_{i},j'_{i-1})-X^{(0)}_{\lfloor t'\ell \rfloor-1-I'+i}\biggr)\\
&\ge \max_{\frac{I}{\ell}\le t\le t'\le 1-\frac{I}{\ell}}\Bigl(\sqrt{\ell}\sigma (W_{t'}-W_{t})\Bigr) 
   -\sum_{i=0}^{I-1}X_{\tilde{t}+i}^{(0)}    - \sum_{i=1}^{I'}X^{(0)}_{\tilde{t}'-I'+i}\\
&\hspace*{4cm}   +\sum_{i=0}^{I-1}\log A_{\tilde{t}+i} (j_{i+1},j_{i})\,
   + \,\sum_{i=1}^{I'}\log A_{\tilde{t}'-I'+i}(j'_{i},j'_{i-1})\Bigr)
\end{align*}
where we define random variables
$\tilde{t}'=\lfloor t'\ell \rfloor-1$ and $\tilde{t}=\lfloor t\ell \rfloor +1$, where
$t'$ and $t$ are the locations at which $W_{t'}-W_{t}$ is a maximum.
We then have
\begin{equation} \label{E:Zbound}
\begin{split}
\E\bigl[ Z(0,\ell-1)\bigr]&\ge \E\left[\max_{\frac{I}{\ell}\le t\le t'\le 1-\frac{I}{\ell}}\Bigl(\sqrt{\ell}\sigma (W_{t'}-W_{t})\Bigr) \right] -\kappa_{1}R\\
&\qquad   +\operatorname{ess}\inf\E\biggl[\sum_{i=0}^{I-1}\log A_{t+i} (j_{i+1},j_{i})
   \, \biggm| \, \eu{D} \biggr]\,
   +\, \operatorname{ess}\inf\E\biggl[\sum_{i=0}^{I'-1}\log A_{t'+i} (j_{i+1},j_{i})
   \, \biggm| \, \eu{D} \biggr]\\
&\ge \E\left[\max_{\frac{I}{\ell}\le t\le t'\le 1-\frac{I}{\ell}}\Bigl(\sqrt{\ell}\sigma (W_{t'}-W_{t})\Bigr) \right]+A_{*}-\kappa_{1}R.
\end{split}
\end{equation}
Note that the essential infimum is required in the second line because $t$ and $t'$ depend on
$\eu{D}$.

By Donsker's invariance principle ({\em cf.} Theorem 2.4.4 of \cite{EKM97})
$(W_{t})_{0\le t\le 1}$ converges weakly (in supremum) to a Brownian motion 
$(\omega_{t})_{0\le t\le 1}$, so that
$$
\E\left[\max_{\frac{I}{\ell}\le t\le t'\le 1-\frac{I}{\ell}}\Bigl(\sigma (W_{t'}-W_{t})\Bigr) \right]
\convinfty{\ell} \sigma\E\left[\max_{0\le t\le t'\le 1}\omega_{t'}-\omega_{t}\right].
$$
The maximum on the right-hand side is sometimes referred to as a ``descent variable'' of the
Brownian motion.\footnote{In \cite{DSY00} it is referred to as the ``downfall'',
but this is an awkward translation from the Russian. In mathematical finance it is referred to as ``maximum drawdown'' \cite{MAPA04}.} By Theorem 2 of \cite{DSY00} it follows that
\begin{equation} \label{E:drawdown}
\lim_{\ell\to\infty} \E\left[\max_{\frac{I}{\ell}\le t\le t'\le 1-\frac{I}{\ell}}\Bigl(\sigma (W_{t'}-W_{t})\Bigr) \right] \ge \sigma\E\bigl[ \max_{0\le t \le 1} (\max_{t\le t'\le 1}\omega_{t'}-\omega_{t})\bigr]=  \sqrt{8/\pi}\log 2 \cdot \sigma.
\end{equation}

Suppose now we do not necessarily have a bound on $X_{t}^{(0)}$.
For any $R>1$,
\begin{align*}
\E\bigl[ (X_{t}^{(0)}-R)\indic_{\{X_{t}^{(0)}\ge R\}} \bigr]& \le \frac{2\tau^{(0)}}{R}e^{-R^{2}/2\tau^{(0)}}\\
\E\bigl[ (X_{t}^{(0)2}-R^{2})\indic_{\{X_{t}^{(0)}\ge R\}}\bigr]& \le 2\tau^{(0)}e^{-R^{2}/2\tau^{(0)}}.
\end{align*}
We now create new random variables $Z^{R}(0,\ell-1)$ 
by replacing every appearance of $X_{t}^{(0)}$ in the definition of $Z(0,\ell-1)$ by
$$
X_{t,R}:={R\wedge X_{t}^{(0)}+\E\bigl[ (X_{t}^{(0)}-R)\indic_{\{X_{t}^{(0)}\ge R\}}\bigr]},
$$
where $R$ may now depend on $\ell$. (This changes $\tX_{t}^{(1)}$ as well as $X_{t}^{(0)}$.) 
We still have $\E[\wt{X}_{t}^{(1)}]=0$, and
Thus the variance
$\sigma_{R}^{2}$ of the new $\tX_{t,R}$ differs from $\sigma^{2}$ by no more than
\begin{align*}
\bigl|\E[(X_{t}^{(0)}&)^{2}]-\E[X_{t,R}^{2}]\bigr|=
  \bigl|\E[(X_{t}^{(0)})^{2}-X_{t,R}^{2}]\bigr| \\
  &\le \E\left[\bigl((X_{t}^{(0)})^{2}-R^{2}\bigr)\indic_{\{X_{t}^{(0)}\ge R\}}\right]
  +2R\E\left[ (X_{t}^{(0)}-R)\indic_{\{X_{t}^{(0)}\ge R\}}\right]
  +\E\left[ (X_{t}^{(0)}-R)\indic_{\{X_{t}^{(0)}\ge R\}}\right]^{2} \\
  &\le \bigl(6\tau^{(0)}+4(\tau^{(0)})^{2}\bigr) e^{-R^{2}/2\tau^{(0)}}
\end{align*}

We see that
$$
\left| Z(0,\ell-1)-Z^{R}(0,\ell-1)\right|\le
  \sum_{t=0}^{\ell-1} \left| X_{t}^{(0)}
  -\left( R\wedge X_{t}^{(0)}+\E\bigl[ (X_{t}^{(0)}-R)\indic_{\{X_{t}^{(0)}\ge R\}}\right) \right|
$$
so that
$$
\bigl|\E\left[Z(0,\ell-1)\right]-\E\left[Z^{R}(0,\ell-1)\right]\bigr| \le \E\left[\left| Z(0,\ell-1)-Z^{R}(0,\ell-1)\right|\right]
  \le \ell \frac{4\tau^{(0)}}{R}e^{-R^{2}/2\tau^{(0)}}. 
$$

Take $R=\log\ell$.
Applying \eqref{E:Zbound} to $Z^{R}(0,\ell-1)$, we may conclude that
\begin{align*}
\liminf_{\ell\to\infty}\ell^{-1/2}\E\bigl[ Z(0,\ell-1)\bigr] &\ge \liminf_{\ell\to\infty}\sigma_{R}\sqrt{8/\pi}\log 2 -
  \limsup_{\ell\to\infty}\ell^{1/2} \frac{4\tau^{(0)}}{R}e^{-R^{2}/2\tau^{(0)}}\\
  &=\sigma \sqrt{8/\pi} \log 2.
\end{align*}
\end{proof}

\subsection{Proof of the upper bound}  \label{sec:equalupper}
We replace $A_{t}(i,j)$ by $\max_{i',j'}A_{t}(i',j')\vee 1$ for all $i\ne j$. 
This can only increase the value of $a(\epsilon)$,
so it suffices to prove the upper bound under this new condition. Similarly,
the upper bound will only be increased if we add $\tmu^{(j)}$ to each
$X_{t}^{(j)}$, so it will suffice to prove the upper bound under the assumption
that all $\tmu^{(j)}$ are 0. We also write 
$$
C_{0}:=\left(5\max_{j} \ttau^{(j)}\right)^{1/2}.
$$
Note that each $Z_{f}$ is sub-Gaussian, and for any $f,f'\in\F_{T}$
$$
Z_{f}-Z_{f'}=\sum_{i=0}^{n-1}  (\tX^{(f_{i})}-\tX^{(f'_{i})})-(K(f)-K(f'))\log\epsilon^{-1}.
$$
The first term is a sum of independent sub-Gaussian random variables with variance factor no
bigger than $2\max_{j} \ttau^{(j)}$, and the second term is deterministic,
so that by Lemma \ref{L:Orlicztau}
$$
\|Z_{f}-Z_{f'}\|_{\Psi} \le d(f,f'):=C_{0} \|f-f'\|_{2}.
$$

We now fix an increasing sequence of integers $1=m_{0}<m_{1}<\cdots<m_{J}<m_{J+1}=T$,
to be determined later, where we assume that $m_{J}=\lfloor T/2\rfloor$.
We define for $J\ge j\ge 0$,
\begin{equation}  \label{E:alphadef}
Z_{*}^{j}:=\max\{Z_{f}\, : \, f\in \bigcup_{m_{j}\le k< m_{j+1}} \F_{T,k} \}.
\end{equation}
We may then use \eqref{E:allterms2} to obtain
\begin{equation}  \label{E:Ynbound1}
a(\epsilon)-a(0) \le 
  \limsup_{T\to\infty} T^{-1}\log \Bigl( 1+\sum_{j=0}^{J-1} \epsilon^{m_{j}} (m_{j+1}-m_{j})\binom{T}{m_{j+1}} \e^{Z^*_{j}}  +T\epsilon^{m_{J}} \binom{T}{m_{J}} \e^{Z^*_{J}} \Bigr) .
\end{equation}
Note that, in order to obtain an upper bound, we need to take a {\em smaller} power of $\epsilon$ than any
in the range $[m_{j},m_{j+1})$ that are combined into one term.

To bound the Orlicz norm of $Z_{*}^{j}$ we use {\em chaining}, as described in
\cite{dP90}. By Lemma 3.4 of \cite{dP90} we know that for any $\F_{*}\subset \F_n$,
\begin{equation}  \label{E:Zstarorlicz}
\|\max_{f\in\F_{*}}Z_{f}\|_{\Psi}\le \sum_{i=1}^{\infty} \frac{C_{0}\sqrt{T}}{2^{i}} \sqrt{2+\log D(C_{0}\sqrt{T}/2^{i},\F_{*})},
\end{equation}
where the packing number $D(r,\F_{*})$ is the maximum number of points
that may be selected from $\F_{*}$, with no two of them having $\|\cdot\|_{2}$ distance smaller than $r/C_{0}$.
(In principle there would be an additional term for the norm of $Z_{f^{(0)}}$, but that is identically 0.)

The packing numbers for $\F_{T,k}$ are difficult to estimate precisely, particularly for large $k$,
but fortunately we can make do with fairly crude bounds, such as we state below as Lemma \ref{L:packingbound}. 
Substituting \eqref{E:packingbound} into \eqref{E:Zstarorlicz}, and using the fact that the bound is increasing in $k$, we see that for $j\le J-1$,
$$
D\Biggl( C_{0}\sqrt{T}/2^{i}\;,\bigcup_{m_{j}< k\le m_{j+1}}\F_{T,k} \Biggr)\le 
	T m_{j+1} d^{m_{j+1}+1}\min\left\{ \frac{T\e}{m_{j+1}}, \frac{T\e}{T/4^{i}-m_{j+1}}+2\e \right\}^{m_{j+1}}
$$
and
$$
D\Biggl( C_{0}\sqrt{T}/2^{i}\;,\bigcup_{m_{J}< k\le T}\F_{T,k} \Biggr)\le 
	T^{2} (8\e d \log d)^{T}
$$

Consider some fixed $j\le J-1$.
If we let $i_{*}=\lfloor\log_{4}T/2m_{j+1}\rfloor$, then for $i\le i_{*}-1,$
$$
\frac{4^{i}m_{j+1}}{T}\le\frac12, \quad\text{ and } \quad 
   2^{i_{*}}\ge \frac12 \sqrt{T/m_{j+1}}\ge \frac{1}{2d}\sqrt{\log(Td\e/m_{j+1})}.
$$

So for $j\le J-1$
\begin{align*}
\|Z_{*}^{j}\|_{\Psi}&\le \sum_{i=1}^{i_{*}-1} \frac{C_{0}\sqrt{T}}{2^{i}} \sqrt{\log 8m_{j+1}Td+m_{j+1}\log Td\e/(4^{-i}T-m_{j+1})} \\
 &\qquad + \sum_{i=i_{*}}^{\infty} \frac{C_{0}\sqrt{T}}{2^{i}} \sqrt{\log 8m_{j+1}Td+m_{j+1}\log (Td\e/m_{j+1})}\\
&\le \sum_{i=1}^{i_{*}-1} \frac{C_{0}\sqrt{T}}{2^{i}} \sqrt{\log (8m_{j+1}Td)+im_{j+1}\log 4-m_{j+1}\log (1-4^{i}m_{j+1}/T)} \\
 &\qquad +\frac{C_{0}\sqrt{T}}{2^{i_{*}-1}} \sqrt{\log (8m_{j+1}Td)+m_{j+1}\log (Td\e/m_{j+1})}\\
&\le \sum_{i=1}^{\infty} \frac{C_{1}\sqrt{T}}{2^{i}} \sqrt{im_{j+1}}
   \;+\; \frac{C_{1}\sqrt{T}}{2^{i_{*}-1}} \sqrt{m_{j+1}\log (Td\e/m_{j+1})}\\
&\le C_{2}\sqrt{Tm_{j+1}}
\end{align*}
for some constants $C_{1},C_{2}$.
By choosing $C_{2}$ appropriately we may ensure that this bound holds as well for $j=J$.

By definition of the Orlicz norm, stated as \eqref{E:orlicz}, this means that for $1\le j\le J$,
$$
\E\left[\exp\left\{(Z_{*}^{j})^{2}/C_{2}^{2}Tm_{j+1}\right\}\right]< 5.
$$
Applying Markov's inequality we have
$$
\P\{ Z_{*}^{j} > z\sqrt{Tm_{j+1}}\}\le 5\e^{-z^{2}/C_{2}^{2}}.
$$

Let $z=\max\{1,C_{2}\}$, and for any $T\ge \log^{2}\epsilon^{-1}$ define $A_{z,T}$ to be the event on which $Z_{*}^{j}\le z\sqrt{Tm_{j+1}(j+1)}$ for all $j$. Note that
\begin{equation} \label{E:Azprob}
\P(A_{z,T}^{\complement})\le \sum_{j=1}^{\infty} \e^{-z^{2}j/C_{2}^{2}}\le \left(\e^{z^{2}/C_{2}^{2}}-1\right)^{-1}.
\end{equation}
This bound is smaller than 1, from which it follows that $\P(A_{z,T})>0$.

We now take $m_{j}:=\lfloor 4Tj z^{2}/ \log^{2}\epsilon \rfloor$ as long as this is $<T/2$,
then set $m_{J}=\lfloor T/2\rfloor$ and $m_{J+1}=T$. We note that $J+1\le ( \log^{2}\epsilon)/4z^{2}$.
By the constraint on $z$, we have for $J-1\ge j\ge 1$
$$
\frac{T\e}{m_{j}}\le \frac{\e\log^{2}\epsilon/4}{jz^{2}-\log^{2}\epsilon/4T}\le
  \frac{ \log^{2} \epsilon}{j z^{2}},
$$
so
$$
\binom{T}{m_{j}}\le \left(  \frac{ \log^{2} \epsilon}{j z^{2}} \right)^{4Tj z^{2}/\log^{2}\epsilon}.
$$
Similarly, we have on $A_{z,T}$ the bound 
$$
Z_{*}^{j}\le z\sqrt{Tm_{j+1}(j+1)}\le 2z^{2}(j+1)T/\log\epsilon^{-1}.
$$
Substituting into \eqref{E:allterms2} we see that on the event $A_{z,T}$,
\begin{equation}  \label{E:Ynbound2}
\begin{split}
a(\epsilon)-a(0)\le \liminf_{T\to\infty}T^{-1}\log &\Bigl( 1 +T\sum_{j=1}^{\infty} \epsilon^{4Tjz^{2}/\log^{2}\epsilon} \left(  \frac{ \log^{2} \epsilon}{j z^{2}}
   \right)^{4Tj z^{2}/\log^{2}\epsilon} \e^{2z^{2}(j+1)T/\log\epsilon^{-1}} \\
   &\hspace*{3cm} +T^{2} (4\epsilon)^{T/2} \e^{Tz\sqrt{J+1}}   \Bigr).
\end{split}
\end{equation}
The summand may be bounded by
$$
\exp\biggl\{ -\frac{2Tz^{2}}{\log\epsilon^{-1}} \left(2j-2-\frac{5}{z\sqrt{j}} \right) \biggr\}\le \exp\biggl\{ \frac{T}{\log\epsilon^{-1}} \left( 4z^{2}+10z \right) \biggr\},
$$
while the additional term on the second line is bounded by
$$
T^{2} (16\epsilon)^{T/4}<1.
$$
Thus
\begin{equation}  \label{E:Ynbound3}
\begin{split}
a(\epsilon)-a(0)
&\le \liminf_{T\to\infty}T^{-1}\left(
  \log T +\log \Bigl( 2T^{-1}+ \left(1-\epsilon^{4}\right)^{-1}\exp\biggl\{ \frac{T}{\log\epsilon^{-1}} \left( 4z^{2}+10z \right) \biggr\} \Bigr)\right)\\
  &= \frac{4z^{2}+10z}{\log \epsilon^{-1}} 
\end{split}
\end{equation}
on the event $A_{z,T}$. Since the event has positive probability, and since
$a(\epsilon)-a(0)$ is almost surely constant, the bound holds with probability 1.

\begin{lemma} \label{L:packingbound}
For any $r$ and positive integers $T,k$, with $\sqrt{T}>r>k>0$,
\begin{equation}  \label{E:packingbound}
D(r,\F_{T,k})\le d^{k+1}\min\left\{ \frac{T\e}{k}, \frac{T\e}{(r/C_{0})^{2}-k}+2\e \right\}^{k}.
\end{equation}
\end{lemma}

\begin{proof}
Let $r'=(r/C_{0})^{2}$. Suppose that $r'>k$, let $j=\lfloor r'/k\rfloor$, and let $m=\lceil T/j\rceil$. Let $\F_{*}$ be the set of ordered (non-decreasing) sequences
of length $k$ from $\{0,\dots,m-1\}$, crossed with $\{0,\dots,d-1\}^{k+1}$, 
and define a map $(\phi,\psi):\F_{n,k}\to \F_{*}$ 
by letting $\{f\}_{i}$ be the $i$-th coordinate where $f$ changes, and defining
$$
\phi(f)_{i}= \bigl\lfloor \frac{\{f\}_{i}}{j}  \bigr\rfloor
$$
and $\psi(f)_{i}=f_{\{f\}_{i}}$; that is, the site that $f$ moves to at its $i$-th
change.

If $f$ and $f'$ are two elements of $\F_{n,k}$ with $\phi(f)=\phi(f')$
and $\psi(f)=\psi(f')$, 
then $f_{i}=f'_{i}$ as long as $\lfloor i/j\rfloor\notin \phi(f)$, 
since any $t\notin \phi(f)$ corresponds to a span of $tj,tj+1,\dots,tj+j-1$
where $f_{tj}=f'_{tj}$ (because they started with $f_{0}=f'_{0}$, and the number
of changes in $f_{0},\dots,f_{tj-1}$ is the same as the number
of changes in $f'_{0},\dots,f'_{ti-1}$). Thus
$d(f,f')\le k\dot j \le r'$, meaning that $\|f-f'\|_{2}\le C_{0}\sqrt{r'}=r$. By the
pigeonhole principle, any subset of $\F_{n,k}$ of size greater than $\# \F_{*}$
has points with $\|\cdot\|$-separation no more than $r$. Hence
$$
D(r,\F_{n,k})\le \#\F_{*} = \binom{m+k}{k}.
$$

Combining this with the trivial bound $D(r,\F_{n,k})\le \#\F_{n,k}=\binom{n}{k}$ and the bound
$$
\binom{a}{b}\le \left(\frac{a\e}{b}\right)^{b}
$$
completes the proof.
\end{proof}

\section{Simulations}  \label{sec:simulation}
We illustrate the results with some simulation tests. We want to show
that in the case where two sites have identical mean expected log growth
(or in the diapause case) the changes in $a$ are approximately like
$1/\log\epsilon^{-1}$ for $\epsilon$ close to 0; and that in the case where
the maximum expected log growth rate occurs at only a single site the changes are like a power of
$\epsilon$, and that the power is approximately as predicted.

\subsection{Diapause case} \label{sec:diapausesim}
We begin with a very simple $2\times 2$ example:
$$
M_{t}=
\begin{pmatrix}
 0& B_{t}\\
 S_{t}&0
\end{pmatrix}, \qquad
A_{t}=
\begin{pmatrix}
1&0\\0&0
\end{pmatrix}, \qquad\text{with }
S_{t}\sim \operatorname{Unif}(0.05,0.99),\quad B_{t} \sim \operatorname{Gam}(5,2),
$$
with $S_{t}$ and $B_{t}$ independent. We have $a(0)=\frac12(\E[\log S_{t}]+\E[\log B_{t}])=-0.0193$. 

\begin{table}[ht]
\centering
\begin{tabular}{ccc}
  \toprule
 $\epsilon$&$1/\log\epsilon^{-1}$&$a(\epsilon)$\\
  \midrule
0.500 & 1.443 & 0.305 \\ 
  0.400 & 1.091 & 0.256 \\ 
  0.300 & 0.831 & 0.206 \\ 
  0.200 & 0.621 & 0.153 \\ 
  0.100 & 0.434 & 0.097 \\ 
  0.050 & 0.334 & 0.065 \\ 
  0.010 & 0.217 & 0.028 \\ 
  0.005 & 0.189 & 0.022 \\ 
  0.001 & 0.145 & 0.012 \\ 
  $10^{-4}$ & 0.109 & 0.003 \\ 
   $10^{-5}$ & 0.087 & -0.001 \\ 
  $10^{-6}$ & 0.072 & -0.005 \\ 
  0 & 0.000 & -0.019 \\ 
   \bottomrule
\end{tabular}
\caption{Simulated diapause example} 
\label{T:diapausesim}
\end{table}

The results are tabulated in Table \ref{T:diapausesim}, for values of $\epsilon$ down to $10^{-6}$. In Figure \ref{F:diapausesim} we plot $a(\epsilon)$
against $1/\log \epsilon^{-1}$, and see that for small values of $\epsilon$ the values are very close to a line (with slope approximately $0.2$.

\begin{figure}[ht]
\begin{center}
\includegraphics[width=12cm]{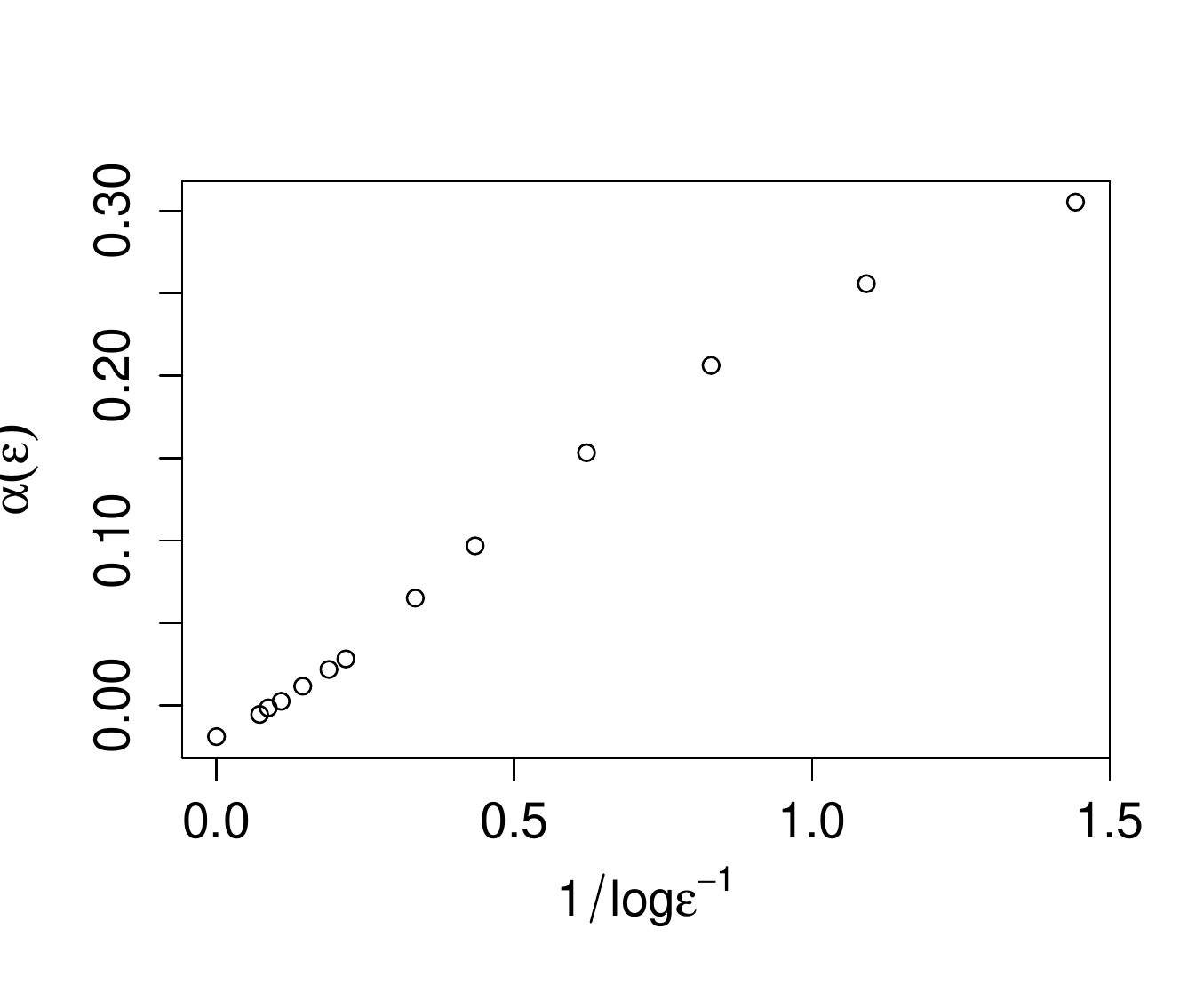}
\caption{Simulated diapause example.}
\label{F:diapausesim}
\end{center}
\end{figure}

\subsection{Migration case}  \label{sec:migrationsim}
We now consider a $3\times 3$ example:
$$
M_{t}(\mu,\sigma^{2})=
\begin{pmatrix}
 e^{\sigma Z_{t}^{(0)}}& 0&0\\
0& e^{\sigma Z_{t}^{(1)}-\mu}&0\\
0&0& e^{\sigma Z_{t}^{(2)}-0.2}
\end{pmatrix}, \qquad
A_{t}(C)=
\begin{pmatrix}
0&C&1\\1&0&C\\C&1&0
\end{pmatrix},
$$
with $Z_{t}^{(0)},Z_{t}^{(1)},Z_{t}^{(2)}$ i.i.d.\  standard normal random variables,
and $C$ is a nonnegative constant. If $C=0$ then the migration graph is
a cycle of length 3, so $\kappa_{1}=\kappa_{2}=3$; if $C>0$ then $\kappa_{1}=\kappa_{2}=2$.

We consider four different cases for $(\mu,\sigma^{2},C)$: $I: (-0.1,0.5,1)$, 
$II: (-0.1,0.5,0)$, $III: (-0.1,1,1)$, and $IV: (0,0.5,1)$. In all cases $a(0)=0$.
The results of this
paper predict that case IV should be like the diapause example in section \ref{sec:diapausesim}, with $a(\epsilon)$ behaving like $c/\log\epsilon^{-1}$
for some constant $c$, when $\epsilon$ is small.

The first three cases should have $\log a(\epsilon)/\log \epsilon^{-1}$
converging to a constant as $\epsilon\downarrow 0$. We have $\tmu^{(1)}=0.1$
in all three cases. For cases I and II we have $\rho^{(1)}=\rho_{*}^{(1)}=0.1$, so
that the power for case I is between
$$
2\cdot 2 \cdot 0.1=0.4 \qquad \text{and}\qquad \frac{2\cdot 2\cdot 0.1}{1+2\cdot 0.1}
=\frac13,
$$
and for case II is between
$$
2\cdot 3 \cdot 0.1=0.6 \qquad \text{and}\qquad \frac{2\cdot 3\cdot 0.1}{1+2\cdot 0.1}
=0.5.
$$
For case III $\rho^{(1)}$ is decreased to 0.05, so the power is between
$$
2\cdot 2 \cdot 0.05=0.2 \qquad \text{and} \qquad \frac{2\cdot 2\cdot 0.05}{1+2\cdot 0.05}=\frac{2}{11}.
$$

We plot some simulated results in Figures \ref{F:diapausesim1} through  \ref{F:diapausesim3},
plotting the $\log a(\epsilon)$ against $\log\epsilon^{-1}$. In the limit as $\epsilon\to 0$
this should approach a line whose slope is in the range given for the power of $\epsilon$ 
in Theorem \ref{T:diffrate1}. We plot lines with those slopes in each figure, and see that
in the lowest range of $\epsilon$ (we take it down to $\epsilon=10^{-6}$) the slope comes down
close to the upper limit, but is still higher. Of course, this is completely consistent with
the true exponent being at the upper limit, particularly since
we don't know anything yet about how
small $\epsilon$ would need to be before the asymptotic slope becomes apparent.

\begin{figure}[ht]
\begin{center}
\includegraphics[width=10cm]{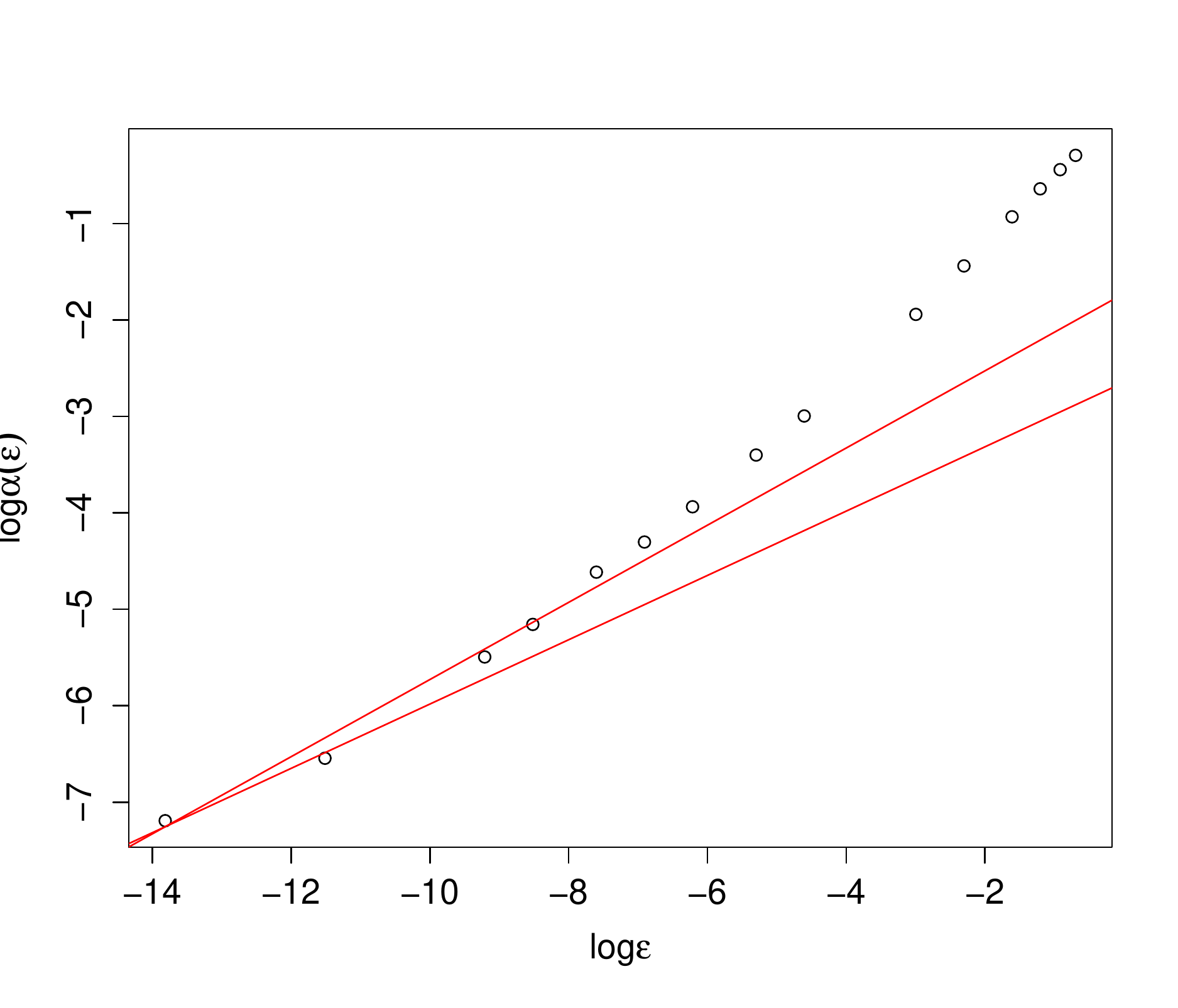}
\caption{Simulated migration example with path length 2. The red lines have slope $0.4$ and $1/3$.}
\label{F:diapausesim1}
\end{center}
\end{figure}

\begin{figure}[ht]
\begin{center}
\includegraphics[width=10cm]{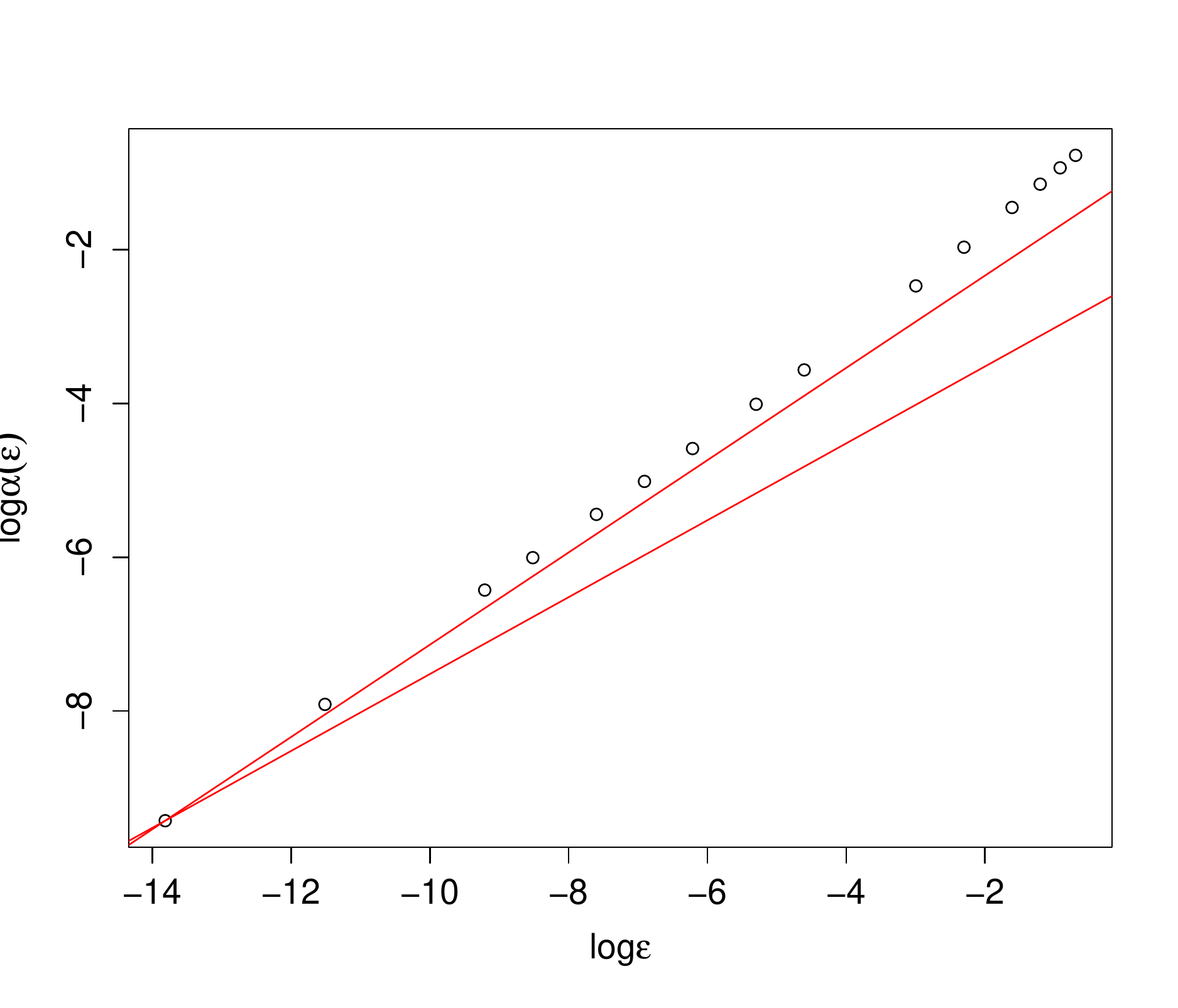}
\caption{Simulated migration example with path length 3. The red lines have slope $0.5$ and $0.6$.}
\label{F:diapausesim2}
\end{center}
\end{figure}

\begin{figure}[ht]
\begin{center}
\includegraphics[width=10cm]{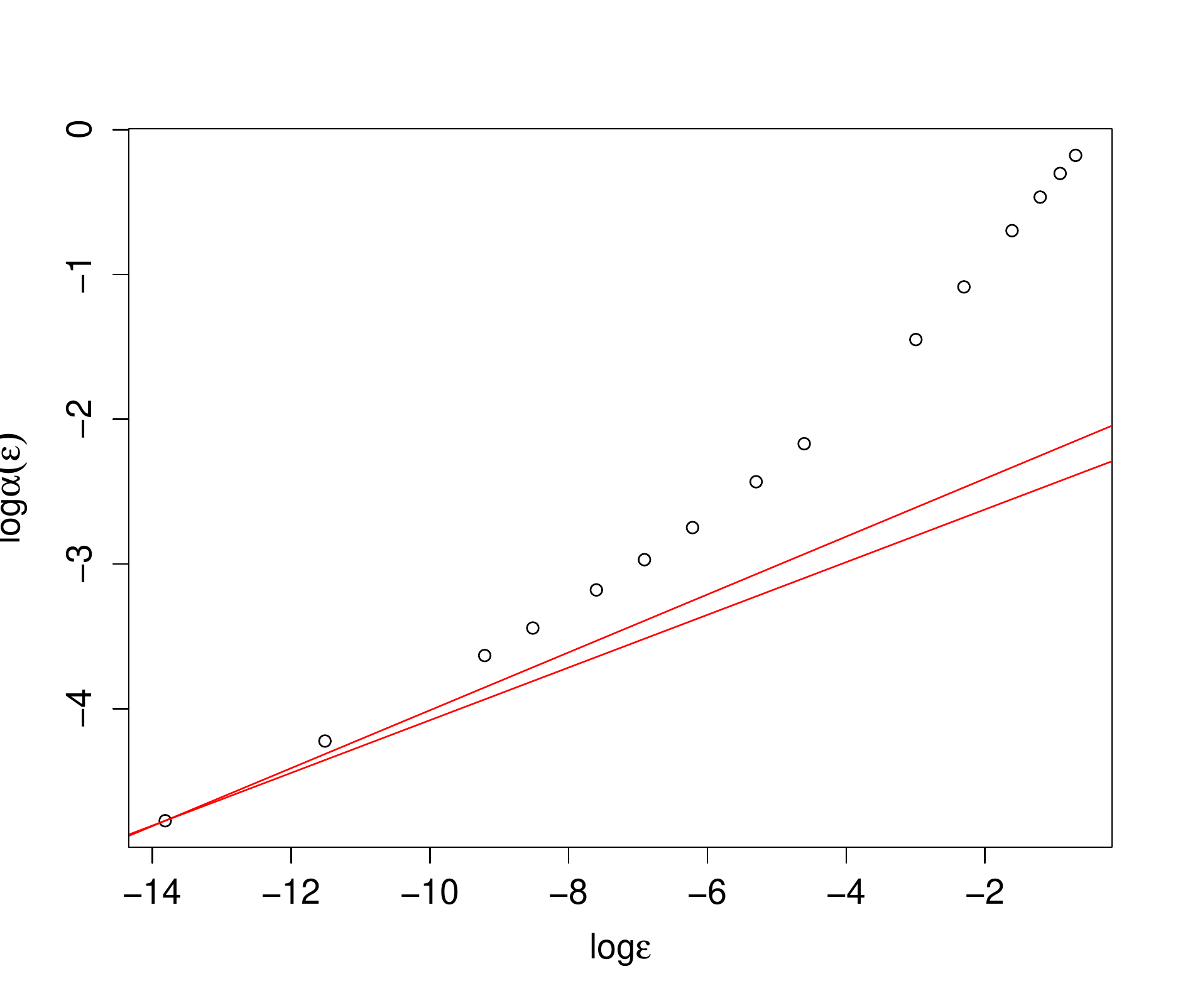}
\caption{Simulated migration example with path length 2, and $\sigma^{2}=1$. The red lines have slope $0.2$ and $2/11$.}
\label{F:diapausesim3}
\end{center}
\end{figure}

\bibliographystyle{alpha}
\bibliography{diapausebib}

\typeout{get arXiv to do 4 passes: Label(s) may have changed. Rerun}

\end{document}